\documentclass[letterpaper]{article}

\usepackage{graphicx} 

\usepackage{fullpage}
\usepackage{authblk} 

\usepackage[round]{natbib}
\usepackage{algorithm}
\usepackage[noend]{algorithmic}

\usepackage{complexity} 
\usepackage{thm-restate} 
\usepackage{amsmath,amsthm,amssymb}
\usepackage{cleveref}
\usepackage{tikz}
    \usetikzlibrary{calc,decorations.pathreplacing}
\usepackage{thm-restate} 
\usepackage{tabularx, booktabs}
\usepackage{nicefrac}
\usepackage{colortbl}
\allowdisplaybreaks 

\usepackage{todonotes} 
\presetkeys{todonotes}{inline}{}

 \newtheorem{theorem}{Theorem}[section]
 \newtheorem{proposition}[theorem]{Proposition}

 \newtheorem{corollary}[theorem]{Corollary}
\newtheorem{claim}{Claim}

\newcommand{\pr}{\mathbb P}

\newcommand{\alg}{\mathit{ALG}}

\usepackage{xcolor}
\definecolor{TUMBlue}{HTML}{0065BD}
\definecolor{TUMGray}{HTML}{808080}
\definecolor{TUMGreen}{HTML}{A2AD00} 

\usepackage{pifont}%
\newcommand{\cmark}{\ding{51}}%
\newcommand{\xmark}{\ding{55}}%
\begin{document}

\title{Stability in Online Coalition Formation}

\author[1]{Martin Bullinger}
\author[2]{Ren{\'e} Romen}

\affil[1]{\small Department of Computer Science, University of Oxford}
\affil[2]{\small School of Computation, Information and Technology, Technical University of Munich
\protect\\ \vspace*{0.05cm}
martin.bullinger@cs.ox.ac.uk, rene.romen@tum.de}

\date{}

\maketitle

\begin{abstract}
Coalition formation is concerned with the question of how to partition a set of agents into disjoint coalitions according to their preferences.
Deviating from most of the previous work, we consider an online variant of the problem, where agents arrive in sequence.
Whenever an agent arrives, they must be assigned to a coalition immediately and irrevocably.
The scarce existing literature on online coalition formation has focused on maximizing social welfare, a demanding requirement, even in the offline setting.
Instead, we seek to achieve \emph{stable} coalition structures online and treat the most common stability concepts based on deviations by single agents and groups of agents.
We present a comprehensive picture in additively separable hedonic games, leading to dichotomies, where positive results are obtained by deterministic algorithms and negative results even hold for randomized algorithms. 
\end{abstract}

\section{Introduction}\label{sec:intro}

The formation of stable coalition structures is an important concern in multi-agent systems.
The critical question is how to partition a set of agents into reasonable coalitions.
A standard framework for this is the consideration of hedonic games \citep{DrGr80a}.
In these games, a set of agents expresses their preferences over subsets of agents containing themselves, i.e., their potential coalitions.
The output is a coalition structure where all agents are assigned to a unique coalition.
In our work, we consider additively separable hedonic games, one of the most prominent classes of hedonic games, where cardinal utilities for single agents encode the preferences, and a sum-based aggregation defines the utility for a coalition.

Hedonic games have been used to model various aspects of group interaction, such as the formation of research teams \citep{AlRe04a} or the detection of communities \citep{ABB+17a}.
A commonality of most research on hedonic games is that the focus is on a single game, which is fully specified and for which a desirable outcome is searched.
However, this misses an important feature of many real-life scenarios:
Agents might arrive over time and have to be assigned to existing coalitions.
For instance, in a company, most workers are already assigned to a department or team, and new hires join existing teams.

Based on such considerations, \citet{FMM+21a} introduced an online variant of hedonic games that adds the arrival of agents over time.
Their goal is to achieve high social welfare, and they find that a greedy algorithm performs best in a competitive analysis.
However, this algorithm has an unbounded competitive ratio if the number of agents or the utility range is unbounded. 
In subsequent work, \citet{BuRo23a} show that it is possible to get rid of utility dependencies of the competitive ratio 
under certain model assumptions, e.g., a uniformly random arrival of agents.
They achieve a competitive ratio of $\Theta\left(n\right)$, which is essentially the best approximation guarantee that we can hope for by efficient algorithms because, for every $\epsilon > 0$, it is \NP-complete to approximate social welfare by a factor of ${n^{1-\epsilon}}$ \cite[Theorem~17]{FKV22a}.\footnote{This result even holds for the class of aversion-to-enemies games that we will introduce and investigate later.}

By contrast, the scarce existing literature on online coalition formation 
omits other common objectives in coalition formation.
Stability probably is the most studied solution concept for hedonic games in general and in additively separable hedonic games in particular \citep[see, e.g.,][]{BoJa02a,SuDi10a,ABS11c,Woeg13a,GaSa19a,BBT22a,BBW21b,Bull22a}.
In our work, we close this research gap and consider the question of whether notions of stability can be achieved in an online manner.

We cover a broad range of the most common stability concepts for hedonic games based on deviations by single agents and groups of agents.
Specifically, we consider Nash stability, individual stability, contractual Nash stability, contractual individual stability as well as the core and the strict core.
In addition, we study Pareto optimality, which is particularly interesting because it is a natural weakening of the demanding objective of maximizing social welfare while it can still be interpreted as a notion of stability \citep{Morr10a}.
We present a comprehensive picture of the capabilities and impossibilities of online algorithms aiming to compute stable partitions.
Note that for some of the solution concepts, e.g., Nash stability, no partition satisfies the stability notion in some instances.
This also implies trivial impossibilities for the online setting.
Therefore, we consider natural utility restrictions, such as symmetry or the distinction of friends and enemies, a natural approach that has been thoroughly explored in coalition formation settings \citep{DBHS06a,OBI+17a,BBT22a,FKV22a,KNR+22a}.
Within the framework of additively separable hedonic games, this is in particular captured in the subclasses of appreciation-of-friends games (AFGs) and aversion-to-enemies games (AEGs) \citep{DBHS06a}.
Our findings are summarized in \Cref{tab:overview}.

\begin{table}
	\caption{Computability of stable 
	partitions by online algorithms; for definitions of stability concepts and utility restrictions, see \Cref{sec:prelims}. A checkmark (\cmark) means a deterministic online algorithm can compute the desired partition. 
	A cross (\xmark) means that no randomized online algorithm exists that outputs the desired partition with probability bounded away from $0$.
	All negative results hold even for the case of symmetric games.
	Of the positive results (highlighted in gray), only the results for contractual Nash stability need symmetry.}
	\label{tab:overview}
\centering
\resizebox{1\textwidth}{!}{
	\renewcommand{\arraystretch}{1.2}
\begin{tabular}{l m{1.75cm} m{1.75cm} m{1.75cm} m{1.75cm} m{1.75cm}}
\toprule
Utility restriction of ASHG & strict & FENG & FEG & AFG & AEG \\ [0.3ex]
Allowed utility values & $\mathbb{Q} \setminus \{0\}$ & $\{1,0,-1\}$ & $\{1, -1\}$ & $\{n, -1\}$ & $\{1,-n\}$ \\ [0.3ex] 
\midrule
Nash stability   & \xmark (Th.~\ref{thm:noISgen}) & \xmark (Th.~\ref{thm:noISgen}) & \xmark (Th.~\ref{thm:noISgen}) & \xmark (Th.~\ref{thm:noISgen}) & \xmark (Th.~\ref{thm:noISgen}) \\
Individual stability   & \xmark (Th.~\ref{thm:noISgen}) & \xmark (Th.~\ref{thm:noISgen}) & \xmark (Th.~\ref{thm:noISgen}) & \xmark (Th.~\ref{thm:noISgen}) & \xmark (Th.~\ref{thm:noISgen}) \\
Contractual Nash stability  & \xmark (Th.~\ref{thm:CNSImpossible}) & \xmark (Th.~\ref{thm:CISImpossible}) & \cellcolor{gray!10} \cmark (Th.~\ref{thm:AlgCNS}) & \xmark (Th.~\ref{thm:CNSImpossible}) &\cellcolor{gray!10} \cmark (Th.~\ref{thm:AlgCNS})\\
Contractual individual stability  & \cellcolor{gray!10} \cmark (Cor.~\ref{cor:strictCIS})& \xmark (Th.~\ref{thm:CISImpossible}) & \cellcolor{gray!10} \cmark (Cor.~\ref{cor:strictCIS}) & \cellcolor{gray!10} \cmark (Cor.~\ref{cor:strictCIS}) & \cellcolor{gray!10} \cmark (Cor.~\ref{cor:strictCIS})\\
Strict core stability  & \xmark (Th.~\ref{thm:core}) & \xmark (Th.~\ref{thm:core}) & \xmark (Th.~\ref{thm:core}) & \xmark (Th.~\ref{thm:coreAFG}) & \xmark (Th.~\ref{thm:core}) \\
Core stability  & \xmark (Th.~\ref{thm:core}) & \xmark (Th.~\ref{thm:core}) & \xmark (Th.~\ref{thm:core}) & \xmark (Th.~\ref{thm:coreAFG}) & \xmark (Th.~\ref{thm:core}) \\
Pareto optimality   & \cellcolor{gray!10} \cmark (Th.~\ref{thm:algPOstrict})& \xmark (Th.~\ref{thm:CISImpossible}) & \cellcolor{gray!10} \cmark (Th.~\ref{thm:algPOstrict}) & \cellcolor{gray!10} \cmark (Th.~\ref{thm:algPOstrict}) & \cellcolor{gray!10} \cmark (Th.~\ref{thm:algPOstrict})\\
\bottomrule
\end{tabular}
} 
\end{table}

A large part of our results is negative and proves the nonexistence of randomized algorithms capable of computing stable coalition structures under strong utility restrictions.
This is even the case for solution concepts like Pareto optimality or contractual individual stability, for which solutions are guaranteed to exist in any hedonic game \citep{AzSa15a}. 
In fact, all our negative results only use games in which coalition structures satisfying the considered solution concept are guaranteed to exist.
By contrast, we obtain deterministic online algorithms capable of computing contractually Nash-stable and Pareto-optimal coalition structures in restricted classes of games. 
While such positive results seem rare, they entail very strong stability guarantees.
The associated algorithms do not only output a final stable coalition structure, but they \emph{maintain} stability throughout the entire arrival process of agents. Otherwise, they would fail their promised guarantee on a partial instance.
Hence, these algorithms suit every application with an indefinite time horizon, where new agents can arrive continuously.

\section{Related Work}

Our work contributes to two streams of work: the consideration of coalition formation models in economic theory and, more recently, the AI literature as well as the investigation of online algorithms in related settings, mostly in theoretical computer science.
Here, we give an account of both.

Coalition formation in the framework of hedonic games was first studied by \citet{DrGr80a} and popularized two decades later \citep{BKS01a,CeRo01a,BoJa02a}.
\citet{BoJa02a} introduced additively separable hedonic games (ASHGs), which have since been an ongoing object of study.
\citet{AzSa15a} present a survey of this stream of work.
The majority of the research on ASHGs considers the offline setting.
It focuses on the computational complexity of stability concepts \citep{DBHS06a,Olse09a,SuDi10a,ABS11c,Woeg13a,GaSa19a,FMM+21a,BBT22a,Bull22a}, but some more recent studies also consider economic efficiency in the sense of Pareto optimality \citep{EFF20a,Bull19a}, or strategyproofness \citep{FKMZ21a}.
Most important to our work, \citet{DBHS06a} and \citet{BBT22a} consider stability in succinct classes of hedonic games based on the distinction of friends and enemies, and the previously cited work settles the complexity of many single-agent stability notions (including all of the notions we consider here) in the offline setting.
Moreover, \citet{Bull19a} presents a polynomial-time algorithm to compute Pareto-optimal partitions for symmetric ASHGs. 

As we mentioned in the introduction, online ASHGs have been introduced by \citet{FMM+21a} and subsequently been studied by \citet{BuRo23a}. 
Moreover, \citet{PSST22a} study online hypergraph matching. 
Their model can be interpreted as coalition formation with bounded coalition sizes.
In contrast to \citet{FMM+21a} and \citet{BuRo23a}, the work by \citet{PSST22a} does not assume additively separable utilities, and agents do not have to be matched immediately at arrival.
However, they depart unmatched after a fixed time.
All three works solely consider the maximization of social welfare or the minimization of total cost.

In addition, a recent series of work considers deviation dynamics for hedonic games, which constitute another time-dependent model of coalition formation \citep[see, e.g.,][]{BFFMM18a,BBW21b,CMM19a}.
In particular, ASHGs and close variants are studied in depth \citep{BMM22a,BBK23a,BBT22a,BuSu23a}.

From the literature on online algorithms, online matching is the most related to our setting, which originates from the seminal paper by \citet{KVV90a}.
\citet{HuTr22a} survey this line of work.
Matchings can be seen as a variant of hedonic games, where coalitions are restricted to size at most~$2$.
Different to our work, the input instances are bipartite, and only one side of the agents appears online.
The objective in online matching is usually to find a matching of maximum cardinality or weight.
\citet{KVV90a} introduce the famous ranking algorithm, which achieves a competitive ratio of $1-\nicefrac{1}{e}$.
Subsequent work considers related models with edge weights, all agents arriving online, or non-bipartite matching \citep{FKM+09a,WaWo15a,HKT+18a,EFGT22a}.
While it is possible to achieve the competitive ratio of $1-\nicefrac{1}{e}$ in the weighted setting \citep{FKM+09a}, this is usually impossible if all agents arrive online \citep{WaWo15a,HKT+18a}.
Closest to our setting is the work by \citet{EFGT22a}.
They show an optimal bound of $\nicefrac{5}{12}$ for maximum weight matching in general graphs in the online random arrival setting and provide an algorithm that matches this bound asymptotically.

Additionally, stability has been considered for online matching to some extent.
\citet{Dova22a} extends stability according to \citet{GaSh62a} to an online setting and shows that her extension can always be satisfied.
Still, her model has several conceptual differences from the online models discussed thus far.
Most notably, agents do not have to be matched immediately (but suffer from a discount in utility if matched later) and may arrive in batches.
Moreover, \citet{GLMV19a} study a two-stage process for school choice with the goal of preserving stability.
While this is generally impossible, they present efficient algorithms that maximize the number of additionally matched agents or minimize the number of reallocations compared to the matching of the first stage.
More loosely related, \citet{BeSa23a} touch upon an online model for recommender systems, where they---mostly experimentally---investigate stability.

\section{Preliminaries}\label{sec:prelims}

In this section, we present preliminaries.
For an integer $i\in \mathbb N$, we define $[i] := \{1,\dots,i\}$.

\subsection{Additively Separable Hedonic Games}
Let $N$ be a finite set of $n$ \emph{agents}. 
Any subset of $N$ is called a \emph{coalition}. 
We denote the set of all possible coalitions containing agent $i\in N$ by $\mathcal N_i := \{C \subseteq N \colon i \in C\}$.
A \emph{coalition structure} (or \emph{partition}) is a partition of the agents.
Given an agent $i \in N$ and a partition~$\pi$, let $\pi(i)$ denote the coalition of $i$, i.e., the unique coalition $C \in \pi$ with $i \in C$.

A \emph{hedonic game} is a pair $(N,\succsim)$ consisting of a set $N$ of agents 
and a preference profile ${\succsim} = (\succsim_i)_{i\in N}$ where $\succsim_i$ 
is a weak order over $\mathcal N_i$ which represents the preferences of agent $i$.
The semantics is that agent~$i$ strictly prefers coalition $C$ to coalition $C'$ if
$C\succ_i C'$ and is indifferent between coalitions $C$ and $C'$ if
$C\sim_i C'$.

An \emph{additively separable hedonic game} (ASHG) consists of a set $N$ of agents and a tuple $u = (u_i)_{i\in N}$ of \emph{utility functions} $u_i\colon N\to \mathbb Q$, such that, for every pair of coalitions $C,C'\in \mathcal N_i$, it holds that $C\succsim_i C'$ if and only if $\sum_{j\in C}u_i(j) \ge \sum_{j\in C'}u_i(j)$ \citep{BoJa02a}.
We usually represent an ASHG by the pair $(N,u)$. 
Clearly, an ASHG is a hedonic game.
We abuse notation and extend the definition of $u$ to coalitions $C\in \mathcal N_i$ and partitions~$\pi$ by $u_i(C) := \sum_{j\in C}u_i(j)$ and $u_i(\pi) := u_i(\pi(i))$, respectively.
Also, an ASHG can be represented equivalently by a complete directed graph $G = (N,E)$ with weight $u_i(j)$ on arc $(i,j)$.
An ASHG is said to be \emph{symmetric} if, for every pair of agents $i,j\in N$, it holds that $u_i(j) = u_j(i)$.
In this case, we also write $u(i,j) := u_i(j)$.
A complete undirected graph can naturally represent a symmetric ASHG.
Following \citet{ABS11c}, an ASHG is \emph{strict} if, for every pair of agents $i,j\in N$, it holds that $u_i(j)\neq 0$.

There are various important subclasses of ASHGs with restricted utility values.
Given a subset $U\subseteq \mathbb Q$, an ASHG is called a \emph{$U$-ASHG} if, for every pair of agents $i,j\in N$, it holds that $u_i(j) \in U$.
In particular, some ASHGs allow a natural interpretation in terms of friends and enemies.
A $U$-ASHG is called an \emph{appreciation-of-friends game} (AFG), \emph{aversion-to-enemies game} (AEG), \emph{friends-and-enemies game} (FEG), or \emph{friends-enemies-and-neutrals game} (FENG) if $U = \{n,-1\}$, $U = \{1,-n\}$, $U = \{1,-1\}$, or $U = \{1,0,-1\}$, respectively \citep{DBHS06a,BBT22a}.
In all of these games, the utility for a coalition depends on the distinction of friends and enemies, i.e., players that yield positive and negative utility, respectively.
In FEGs and FENGs, friends and enemies have equal importance, whereas in AFGs, a single friend outweighs an arbitrary number of enemies, and in AEGs, a single enemy annihilates any number of friends.

\subsection{Solution Concepts}\label{sec:solcon}

In this section, we define the solution concepts considered in this paper.
\Cref{fig:solution-concepts} gives an overview of their logical relationships.
We assume that we are given a fixed ASHG $(N,u)$.

\begin{figure}
	\centering
	\begin{tikzpicture}
		\usetikzlibrary{shapes.misc}
		\tikzstyle{roundedbox}=[draw, rounded rectangle, rounded rectangle arc length=180, minimum height=1.8em]
	\node (NS) [roundedbox] at (0, 0) {Nash Stability};
	\node (SCS) [roundedbox] at (5, 0) {Strict Core};
	\node[TUMGray] (WO) [roundedbox] at (10, 0) {Welfare Optimality};
	
	\node (CNS) [roundedbox] at (0, -2) {Contractual Nash Stability};
	\node (IS) [roundedbox] at (4.5, -2) {Individual Stability};
	\node (CS) [roundedbox] at (7.3, -2) {Core};
	\node (PO) [roundedbox] at (10, -2) {Pareto Optimality};
	
	\node (CIS) [roundedbox] at (5, -4) {Contractual Individual Stability};
	
	\draw[->] (NS) edge (CNS);
	\draw[->] (NS) .. controls +(down:1cm) and +(up:1cm) .. (IS);
	\draw[->] (SCS) .. controls +(down:1cm) and +(up:1cm) ..  (IS);
	\draw[->] (SCS) .. controls +(down:1cm) and +(up:1cm) .. (CS);
	\draw[->] (SCS) .. controls +(down:1cm) and +(up:1cm) .. (PO);
	
	\draw[->] (CNS) .. controls +(down:1cm) and +(up:1cm) .. (CIS);
	\draw[->] (IS) .. controls +(down:1cm) and +(up:1cm) .. (CIS);
	\draw[->] (PO) .. controls +(down:1cm) and +(up:1cm) .. (CIS);
	\draw[->] (WO) .. controls +(down:1cm) and +(up:1cm) .. (PO);

	\end{tikzpicture}
	
	\caption{Logical relationships between our solution concepts \citep[see, e.g.,]{AzSa15a}. 
	An arrow from concept $\alpha$ to concept $\beta$ indicates that if a partition satisfies~$\alpha$, then it also satisfies~$\beta$. 
	For reference, we also depict welfare optimality.
	}
	\label{fig:solution-concepts}
\end{figure}

Notions of stability capture agents' incentives to perform deviations \citep{BoJa02a,DiSu07a}.
A \emph{single-agent deviation} performed by agent~$i$ 
transforms a partition $\pi$ into a partition $\pi'$ 
where $\pi(i)\neq\pi'(i)$ and, for all agents $j\neq i$, it holds that $\pi(j)\setminus\{i\} = \pi'(j)\setminus\{i\}$.
The basic idea of deviations is that the deviating agent should immediately benefit from a deviation. 
A \emph{Nash deviation} is a single-agent deviation performed by agent~$i$ such that $u_i(\pi') > u_i(\pi)$.
Any partition in which no Nash deviation is possible is said to be \emph{Nash-stable} (NS).

The drawback of Nash stability is that only the preferences of the deviating agent are considered, which might seem too demanding in the context of cooperation.
Therefore, various refinements have been proposed, which additionally require the consent of the abandoned or the welcoming coalition.
An \emph{individual deviation} (or \emph{contractual deviation}) is a Nash deviation by agent~$i$ transforming $\pi$ into $\pi'$ 
such that, for all agents $j\in \pi'(i)\setminus\{i\}$ (or $j\in \pi(i)\setminus\{i\}$), it holds that $u_j(\pi')\ge u_j(\pi)$. 
Then, a partition is said to be \emph{individually stable} (IS) or \emph{contractually Nash-stable} (CNS) if it allows for no individual or contractual deviation, respectively. 
A single-agent deviation is called a \emph{contractual individual deviation} if it is both a contractual deviation and an individual deviation.
A partition is said to be \emph{contractually individually stable} (CIS) if it allows no contractual individual deviation.

Next, we introduce stability based on group deviations.
Consider a partition $\pi$ and a coalition $B\subseteq N$.
Then, $B$ is called a \emph{blocking coalition} for $\pi$ if, for all agents $i\in B$, it holds that $u_i(B) > u_i(\pi)$.
Moreover, $B$ is called a \emph{weakly blocking coalition} for $\pi$ if, for all agents $i\in B$, it holds that $u_i(B) \ge u_i(\pi)$, and there exists an agent $j\in B$ with $u_j(B) > u_j(\pi)$.
A partition is said to be in the \emph{core} (CR) (or \emph{strict core} (SCR)) if it admits no blocking coalition (or weakly blocking coalition).
Note that every blocking coalition is also weakly blocking.
Hence, the strict core prevents a larger set of possible group deviations and, therefore, is the stronger solution concept.
Some authors refer to the strict core as strong core \citep{BoJa02a}.
When we speak of performing a group deviation, we mean that agents form a (weakly) blocking coalition.

For a more concise notation, we refer to deviations with respect to stability concept $\alpha\in \{$NS, IS, CNS, CIS, CR, SCR$\}$ as \emph{$\alpha$ deviations}, e.g., IS deviations for $\alpha =$ IS.
Similarly, we refer to a partition satisfying stability concept~$\alpha$ as an $\alpha$ partition.

Finally, we also consider Pareto optimality, which can be seen as a stability guarantee, where the whole set of agents cannot perform a group deviation.
A partition $\pi'$ is said to \emph{Pareto-dominate} a partition $\pi$ if, for every agent $i\in N$, it holds that $u_i(\pi')\ge u_i(\pi)$, and there exists an agent $j\in N$ with $u_j(\pi') > u_j(\pi)$.
A partition $\pi$ is said to be \emph{Pareto-optimal} (PO) if it is not Pareto-dominated by another partition.
Pareto optimality is a classical concept of economic efficiency, and already was a primal objective during the birth of hedonic games \citep{DrGr80a}.
Note that Pareto optimality is a weakening of welfare optimality, which was the objective in the literature on online ASGHs thus far \citep{FMM+21a,BuRo23a}.
Moreover, at first glance, Pareto dominance feels like a global variant of weakly blocking coalitions, and therefore seemingly leads to a more demanding solution concept.
However, the exact opposite is true.
While group deviations based on (weakly) blocking coalitions lead to partitions that can be inferior for other agents, every Pareto dominance gives rise to a weakly blocking coalition (the ones containing agents that are strictly better off).
Hence, partitions in the strict core are Pareto-optimal.

\subsection{Online Coalition Formation}

In this section, we introduce our model of online coalition formation, following the notation of \citet{BuRo23a}.
The online setting is not restricted to ASHGs, so we define it for a general hedonic game $G = (N,\succsim)$.
Let $\Sigma(G) := \{\sigma\colon [|N|] \to N \textnormal{ bijective}\}$ be the set of all \emph{orders} of the agent set $N$.
Given a subset of agents $M\subseteq N$, let $G[M]$ be the hedonic game restricted to agent set $M$, i.e., the hedonic game $(M,\succsim^M)$ where 
 $C\succsim_i^M D$ if and only if $C\succsim_i D$.
Moreover given a partition $\pi$ of a set $N$ and a subset of agents $M\subseteq N$, we define $\pi[M]$ as the \emph{partition restricted to $M$} as $\pi[M] := \{C\cap M\colon C\in \pi, C\cap M\neq \emptyset\}$.
Specifically, if $M$ consists of all agents except a single agent $i\in N$, then we write $\pi - i := \pi[N\setminus\{i\}]$.

An instance of an \emph{online coalition formation} problem is a pair $(G,\sigma)$, where $G = (N,\succsim)$ is a hedonic game and $\sigma\in \Sigma(G)$.
An \emph{online coalition formation algorithm} for instance $(G,\sigma)$ gets as input the sequence $G_1,\dots,G_n$, where, for every $i\in [n]$, $G_i := G[\{\sigma(j)\colon 1\le j \le i\}]$.
Then, for every $i\in [n]$, the algorithm has to produce a partition $\pi_i$ of $\{\sigma(j)\colon 1\le j \le i\}$ such that
\begin{itemize}
	\item the algorithm has only access to $G_i$ and
	\item $\pi_i - \sigma(i) = \pi_{i-1}$.
\end{itemize}
The output of the algorithm is the partition $\pi_n$.
Given an online coalition formation algorithm $\alg$, let $\alg(G,\sigma)$ be its output for instance $(G,\sigma)$.
If $\sigma$ is clear from the context, we omit it from this notation and simply write $\alg(G)$.

More informally, the algorithm iteratively creates a partition such that it only has access to the utilities of the currently present agents when irrevocably adding a new agent to an existing or new coalition.
In addition to deterministic algorithms, we also consider randomized algorithms. 
This means that the decisions as to which coalition an agent is added to can be random. 

Unlike welfare optimality, stability concepts do not naturally yield a quantitative maximization objective, and we cannot directly perform the usual competitive analysis.
Instead, we have qualitative objectives that are either satisfied or not by an output.
Therefore, we desire algorithms that output stable partitions with a high probability if agents arrive online, which once again is a quantitative objective.

Consider a solution concept $\alpha\in \{$NS, IS, CNS, CIS, PO, CR, SCR$\}$ and an algorithm $\alg$.
We define the \emph{$\alpha$ guarantee} of $\alg$ as
\[
W_\alpha(\alg) :=
\inf_G\min_{\sigma\in\Sigma(G)}\pr(\alg(G,\sigma) \textnormal{ is an }\alpha\textnormal{ partition})\text.
\]
Here, the probability is taken according to the randomized decisions of $\alg$.
Hence, $W_\alpha(\alg)$ is the worst-case probability that $\alg$ outputs an $\alpha$ partition.

Note that the $\alpha$ guarantee of deterministic algorithms is either one---if the algorithm always outputs an $\alpha$ partition---or zero---if the algorithm does not output an $\alpha$ partition for some input instance.

\section{Results}\label{sec:stab}

In this section, we present our results.
Only instances that allow for a stable partition are relevant for the consideration of stability in online coalition formation.
Otherwise, no algorithm, and therefore especially no online algorithm, can make any guarantee.\footnote{\label{fn:onlinevsoffline}We can also take the viewpoint of comparing the capabilities of online algorithms with offline possibilities: if no stable partition exists, then any online algorithm is as good as an optimal offline algorithm.}

In the literature on stability in ASHGs, two restrictions have turned out to be vital for the existence of stable partitions, namely symmetry and severe utility restrictions \citep{BoJa02a,DBHS06a,BBT22a}.
First, in symmetric ASHGs, Nash-stable partitions (and therefore partitions satisfying all weaker single-deviation stability notions) are guaranteed to exist \citep{BoJa02a}.
Moreover, in (possibly asymmetric) FEGs, AEGs, and AFGs, it is guaranteed that individually stable and contractually Nash-stable partitions exist \citep{BBT22a}.
Finally, utility restrictions may also lead to group stability \citep{DBHS06a}:
In particular, AFGs contain partitions in the strict core. 
While this is not the case for AEGs, these at least contain partitions in the core.
In this section, we will see (the conjunction of) which of these assumptions are sufficient to allow for the computation of stable outcomes online, and which of the results in the offline setting cause problems online.

\subsection{Contractual Nash Stability}

We start with the consideration of contractual Nash stability, where every agent in the abandoned coalition can veto a single-agent deviation.
As a warm-up, we begin with a simple proposition that gives a first hint as to why computing stable partitions in an online manner is a nontrivial task. 
Even the conjunction of symmetry and utility restrictions is not sufficient for computing CNS partitions.

\begin{proposition}\label{prop:CNSnotAFG}
	There exists no deterministic online algorithm, which always outputs a CNS partition for symmetric AFGs.
\end{proposition}

\begin{figure}
	\centering
	\begin{tikzpicture}
		
		\node[draw, circle](a1) at (150:1.2){\color{white}{$v3$}};
		\node at (a1){$a$};
		\node[draw, circle](b1) at (30:1.2){\color{white}{$v3$}};
		\node at (b1) {$b$};
		\node[draw, circle](c1) at (270:1.2){\color{white}{$v3$}};
		\node at (c1) {$c$};
		
		\draw (a1) edge node[pos = 0.5, fill = white,yshift=1pt]{$-1$} (b1);
		\draw (a1) edge node[pos = 0.5, fill = white]{$n$}  (c1);
		\draw (b1) edge node[pos = 0.5, fill = white]{$n$}  (c1);

		\node[draw, circle](a2) at ($(5,0)+(150:1.2) $){\color{white}{$v3$}};
		\node at (a2){$a$};        
		\node[draw, circle](b2) at ($(5,0)+(30:1.2)$){\color{white}{$v3$}};
		\node at (b2) {$b$};       
		\node[draw, circle](c2) at ($(5,0)+(270:1.2)$){\color{white}{$v3$}};
		\node at (c2) {$c$};
		
		\draw (a2) edge node[pos = 0.5, fill = white,yshift=1pt]{$-1$} (b2);
		\draw (a2) edge node[pos = 0.5, fill = white]{$-1$} (c2);
		\draw (b2) edge node[pos = 0.5, fill = white]{$-1$} (c2);

		\node[draw, circle](a) at ($(0,2) + (150:1.2) + (barycentric cs:a1=1,b2=1)$){\color{white}{$v3$}};
		\node at (a){$a$};
		\node[draw, circle](b) at ($(0,2) + (30:1.2) + (barycentric cs:a1=1,b2=1)$){\color{white}{$v3$}};
		\node at (b){$b$};
		
		\draw (a) edge node[pos = 0.5, fill = white,yshift=1pt]{$-1$} (b);

		\draw[bend right] ($(a)+(-.5,-.3)$) edge[->] ($(0,.5) + (barycentric cs:a1=1,b1=1)$);
		\draw[bend left] ($(b)+(.5,-.3)$) edge[->] ($(0,.5) + (barycentric cs:a2=1,b2=1)$);
		
	\end{tikzpicture}
	\caption{Adversarial AFGs for computing CNS partitions.
	Every deterministic algorithm fails for one of the two possible instances.}
	\label{fig:noCNSalg}
\end{figure}

\begin{proof}
	Assume for contradiction that $\alg$ always outputs a CNS partition for symmetric AFGs.
	Consider the following two AFGs $(N,u_1)$ and $(N,u_2)$ with identical agent set $N = \{a,b,c\}$ and symmetric utilities $u_1(a,b) = u_2(a,b) = -1$, $u_1(a,c) = u_1(b,c) = n$ (in this case, $n = 3$), and $u_2(a,c) = u_2(b,c) = -1$.
	Consider the arrival order $a$, then $b$, then~$c$.
	Before the arrival of $c$, $\alg$ cannot distinguish, whether the game will be $(N,u_1)$ or $(N,u_2)$.
	The situation is depicted in \Cref{fig:noCNSalg}.
	
	If $\alg$ creates $\{a,b\}$ at the arrival of $b$, then it fails for $(N,u_2)$. 
	Hence, $\alg$ has to create a new coalition $\{b\}$. 
	When $c$ arrives, $\alg$ cannot form a new singleton coalition as otherwise $a$ (or $b$) has a CNS deviation to join $c$.
	Assume without loss of generality that $\alg$ forms $\{a,c\}$.
	Then, $b$ has a CNS deviation to join them.
	Hence, $\alg$ fails for $(N,u_1)$.
\end{proof}

However, the previous result seems to rely on small negative utilities.
In fact, we can compute CNS partitions with an online algorithm for other restricted classes.
The basic idea of our algorithm is to establish that agents in a coalition of size at least~$2$ are not allowed to leave their coalition because some other agent would veto this.
Moreover, we use our assumption on the utility values to show that agents in singleton coalitions can never gain positive utility by joining a constructed coalition.

\begin{theorem}\label{thm:AlgCNS}
	Let $y \ge x > 0$. Then, there exists a deterministic online algorithm, which always outputs a CNS partition for symmetric $\{-y,x\}$-ASHGs.
\end{theorem}

\begin{proof}
	Let $y \ge x > 0$ and consider a symmetric $\{-y,x\}$-ASHG with agent set $N = \{a_i\colon 1\le i \le n\}$ and arrival order $a_1,\dots, a_n$.
	We apply \Cref{alg:CNSsymFEG} to compute a partition $\pi$.
	This algorithm proceeds as follows.
	Whenever a new agent arrives, we compute the set of present agents with positive utility for the new arrival.
	Note that, by symmetry, this implies a mutual positive utility.
	Assume there is at least one.
	Then, we check if at least one of them is in a singleton coalition.
	If such an agent exists, we let the new agent join this singleton coalition.
	Otherwise, we add $a_i$ to any coalition of an agent with positive utility.
	If no such agent exists, we form a new singleton coalition.
	
\begin{algorithm}
  \caption{Contractually Nash-stable partition of online symmetric $\{-y,x\}$-ASHGs for $y\ge x > 0$.}
  \label{alg:CNSsymFEG}
  \begin{flushleft}
    \textbf{Input:} Symmetric $\{-y,x\}$-ASHG\\
    \textbf{Output:} Contractually Nash-stable partition $\pi$
  \end{flushleft}

  \begin{algorithmic}[]
\STATE $\pi\leftarrow \emptyset$
  \FOR{$i = 1,\dots, n$}
  \STATE $N_i \leftarrow \{j\in [i-1]\colon u(a_i,a_j)>0\}$
  \IF {$\exists j\in N_i$ with $|\pi(a_j)| = 1$}
	\STATE $\pi \leftarrow \pi\setminus \{\{a_j\}\}\cup \{\{a_i,a_j\}\}$
	\ELSIF {$\exists j\in N_i$ with $|\pi(a_j)| > 1$}
	\STATE $\pi \leftarrow \pi\setminus \{\pi(a_j)\}\cup \{\pi(a_j)\cup\{a_i\}\}$
	\ELSE
	\STATE $\pi \leftarrow \pi\cup \{\{a_i\}\}$
  \ENDIF
	\ENDFOR
  \RETURN $\pi$
 \end{algorithmic}
\end{algorithm}

	We claim that the partition~$\pi$ is contractually Nash-stable.
	We show the claim by induction.
	Recall that, for every $i\in [n]$, $\pi_i$ is the partition created by the algorithm after agent $a_i$ has been assigned to a coalition.
	\begin{claim}\label{cl:CNS}
		For every $i\in [n]$, the following statements are true:
		\begin{enumerate}
			\item The partition $\pi_i$ is CNS.
			\item For every coalition $C\in \pi_i$ with $|C|\ge 2$ and $a_k\in C$, there exists an agent $a_\ell\in C$ with $u(a_k,a_\ell) = x$.
			\item For every $k\in [i]$ with $|\pi_i(a_k)| = 1$,
			it holds that, if $\ell\in [i]$ with $u(a_k, a_\ell) = x$, then $|\pi_i(a_\ell)|\ge 2$ and $u(a_k, b) = - y$ for all $b\in \pi_i(a_\ell)\setminus\{a_\ell\}$
		\end{enumerate}
	\end{claim}
	\renewcommand\qedsymbol{$\vartriangleleft$}
	\begin{proof}
		All three statements are true for $i = 1$. 
		Assume now that the statements are true for some $1\le i < n$.
		Let $N_{i+1} := \{j\in [i]\colon u(a_{i+1},a_j)>0\}$.
		We start by proving the second and third assertions by a case distinction according to the different cases in the algorithm.
		
		Assume first that there exists $j\in N_{i+1}$ with $|\pi_i(a_j)| = 1$ and that we have $\pi_{i+1} = \pi_i \setminus \{\{a_j\}\} \cup \{\{a_j,a_{i+1}\}\}$. 
		Clearly, the second assertion for $i+1$ follows by induction and $u(a_{i+1},a_j) = x$.
		For the third assertion, let $k\in [i+1]$ with $|\pi_{i+1}(a_k)| = 1$ and consider $\ell\in [i+1]$ with $u(a_k, a_\ell) = x$.
		Since $|\pi_{i+1}(a_{i+1})| = 2$, it holds that $k\neq i+1$.
		By the induction hypothesis for the third assertion, the third assertion is true unless $\ell \in \{i+1,j\}$.
		Moreover, again by the induction hypothesis for the third assertion, $u(a_k,a_j) = -y$, and therefore the assertion is true if $\ell = i+1$.
		
		Next, assume that there exists no $j\in N_{i+1}$ with $|\pi_i(a_j)| = 1$, but $N_{i+1} \neq \emptyset$ and that
		$\pi_{i+1} = \pi_i\setminus \{\pi_i(a_j)\}\cup \{\pi_i(a_j)\cup\{a_{i+1}\}\}$. 
		Note that the second assertion is true for agent $a_{i+1}$ because $a_j\in \pi_{i+1}(a_{i+1})$. 
		For all other agents, the second assertion follows by induction.
		As there exists no $j\in N_{i+1}$ with $|\pi(a_j)| = 1$, the third assertion follows by induction.
		
		Finally, if $N_{i+1} = \emptyset$, then  $\pi_{i+1} = \pi_i \cup \{\{a_{i+1}\}\}$.
		Hence, the second assertion follows by induction, and the third assertion follows by induction for all agents except for $a_{i+1}$.
		For $a_{i+1}$, it is true because $N_{i+1} = \emptyset$.
		
		It remains to prove the first assertion.
		We show how it follows from the second and third assertions.
		By the second assertion, no agent in a coalition of size at least $2$ is allowed to leave their coalition and can, therefore, not perform a CNS deviation.
		On the other hand, by the third assertion, no agent in a singleton coalition can improve their utility by joining any other coalition.
		Hence, $\pi_{i+1}$ is a CNS partition.
		This completes the proof of the claim.
	\end{proof}
	
	The assertion of \Cref{thm:AlgCNS} follows from \Cref{cl:CNS} for the case $i = n$.
\renewcommand\qedsymbol{$\square$}\end{proof}

In particular, \Cref{thm:AlgCNS} applies to symmetric FEGs and AEGs.
For the latter, we must deal with variable utility values that depend on the number of players.
However, \Cref{thm:AlgCNS} applies to individual games, and each AEG satisfies the conditions of the theorem. 
By contrast, \Cref{thm:AlgCNS} breaks down if we additionally allow for the utility value of~$0$, even if the positive and negative utilities are restricted further.
We defer this result to \Cref{sec:CIS}, where we get it as a byproduct during the consideration of CIS.

To conclude this section, we show how to strengthen \Cref{prop:CNSnotAFG} for randomized algorithms.
The idea is to construct a random instance where every deterministic algorithm performs poorly.
We can then apply Yao's principle \citep{Yao77a} to bound the performance of any randomized algorithm.
For this, we create a random version of the game in \Cref{prop:CNSnotAFG}, where every deterministic algorithm succeeds in computing a CNS partition with probability at most $\nicefrac{1}{2}$.
Modifying this instance by concatenating $k$ copies of this instance implies that every deterministic algorithm succeeds with probability at most $2^{-k}$.

\begin{theorem}\label{thm:CNSImpossible}
	Let $\alg$ be any randomized online algorithm for symmetric AFGs.
	Then, it holds that $W_{\textnormal{CNS}}(\alg) = 0$.
\end{theorem}

\begin{proof}
	Let $k\in \mathbb N$ be a positive integer. 
	We define the random AFG $G = (N, u)$ where $N = \bigcup_{i\in [k]}N_i$ for $N_i = \{a_i,b_i,c_i\}$.
	The random utilities are given by
	$u(a_i,b_i) = -1$ and, with probability $\nicefrac 12$ each, it holds that $u(a_i,c_i) = u(b_i,c_i) = 3k$ or  $u(a_i,c_i) = u(b_i,c_i) = -1$. 
	Note that $3k$ is the number of agents.
	All other utilities are set to $-1$.
	The randomizations for the utilities within $N_i$ and $N_j$ for $1\le i < j\le k$ are performed independently.
	The agents arrive in the order $N_1$, \dots, $N_k$ and within a set $N_i$, first $a_i$, then $b_i$, then $c_i$.
	In other words, $G$ is a uniformly random choice from a set of $2^k$ instances, each of which is a composition of $k$ gadgets drawn independently from the same distribution.
	Each gadget is one of the two games considered in \Cref{prop:CNSnotAFG} with equal probability.
	
	Now, let $\alg$ be an arbitrary deterministic algorithm for AFGs and define $\pi := \alg(G)$.
	By the proof of \Cref{prop:CNSnotAFG}, for every $i\in [k]$, $\alg$ fails with probability at least $\nicefrac{1}{2}$ on $G[N_i]$.
	
	Moreover, by design of the random instance, if $\alg$ computes a CNS partition on $G$, then, for all $i\in [k]$, $\pi[N_i]$ is CNS for $G[N_i]$.
	Indeed, if $u(a_i,c_i) = u(b_i,c_i) = -1$, then $\pi$ is CNS only if all agents in $N_i$ are in singleton coalitions, and hence $\pi[N_i]$ is CNS for $G[N_i]$.
	If $u(a_i,c_i) = u(b_i,c_i) = 3k$, then $a_i\in \pi(c_i)$ and $b_i\in \pi(c_i)$ as these agents would perform a CNS deviation to join $c_i$, otherwise. 
	Hence, $\pi[N_i] = \{N_i\}$, which is CNS for $G[N_i]$.

	Now, observe that, by the arrival sequence of the agents, the performance of $\alg$ on $N_i$ is at most as good as the performance of the best algorithm for $G[N_i]$.
	Therefore, using the independence of the random selection of the utilities, the probability that $\alg$ computes a CNS partition on $G$ is bounded by the product of the probabilities that $\pi[N_i]$ is CNS for $G[N_i]$.
	Hence, $\pi$ is a CNS partition with probability at most $2^{-k}$.
	
	By Yao's principle, no randomized algorithm can compute a CNS partition with probability more than $2^{-k}$ for every (deterministic) symmetric AFG.
	Since $k$ is chosen arbitrarily, this proves the assertion.
\end{proof}

\subsection{Contractual Individual Stability and Pareto Optimality}\label{sec:CIS}

Next, we consider CIS, the weakening of CNS, where an agent in the welcoming coalition can also veto a single-agent deviation.
\Cref{alg:CNSsymFEG} can be used to compute CIS partitions, even if we allow for strict and symmetric ASHGs as input.
However, our following result achieves even more and shows the existence of an online algorithm for computing PO partitions in strict (and possibly nonsymmetric) ASHGs.
Recall that PO is a stronger notion than CIS.
The presented algorithm is an online adaptation of serial dictatorships.
This algorithmic approach is known to be successful in achieving Pareto optimality for offline ASHGs \citep{ABS11c,Bull19a} and online fair division \citep{AlWa19a}.\footnote{Similar to the online fair division literature, our online serial dictatorship algorithm can be shown to have the additional desirable property of strategyproofness.}
The idea is to assign a dictator to every created coalition, and these are asked in the order of their arrival whether they want newly arriving agents to be part of their coalition.

\begin{theorem}\label{thm:algPOstrict}
	There exists a deterministic online algorithm, which always outputs a PO partition for strict ASHGs.	
\end{theorem}

\begin{proof}
	Consider a strict ASHG with agent set $N = \{a_i\colon 1\le i \le n\}$ and arrival order $a_1,\dots, a_n$.
	Apply \Cref{alg:POstrictASHG} to compute a partition $\pi$.
	This algorithm proceeds as follows:
	Whenever a new agent arrives, we ask for each existing coalition whether the first agent in that coalition has a positive utility for the new agent.
	If such a coalition exists, the new agent joins the coalition among those that was created first.
	Otherwise, the algorithm starts a new coalition with the new agent.
	
\begin{algorithm}
  \caption{Pareto-optimal partition of online strict ASHG}
  \label{alg:POstrictASHG}
  \begin{flushleft}
    \textbf{Input:} Strict ASHG\\
    \textbf{Output:} Pareto-optimal partition $\pi$
  \end{flushleft}

  \begin{algorithmic}[]
\STATE $\pi\leftarrow \emptyset$, $k\leftarrow 0$
  \FOR{$i = 1,\dots, n$}
  \IF {$\{j\in [k]\colon u_{\ell_j}(a_i)>0\}\neq \emptyset$}
	\STATE $j^* \leftarrow \min_{j\in [k]}\{u_{\ell_j}(a_i)>0\}$
	\STATE $\pi \leftarrow \pi\setminus \{C_{j^*}\}\cup \{C_{j^*}\cup\{a_i\}\}$
  \ELSE
	\STATE $k\leftarrow k+1$
	\STATE $C_k\leftarrow \{a_i\}$, $\ell_k\leftarrow a_i$
	\STATE $\pi\leftarrow \pi\cup \{C_k\}$
  \ENDIF
	\ENDFOR
  \RETURN $\pi$
 \end{algorithmic}
\end{algorithm}

	For the proof, we use the notation from the algorithm and assume that $\pi = \{C_i\colon 1 \le i \le m\}$ for some $m>0$.
	Assume further that these coalitions were formed in the order $C_1,\dots,C_m$ and that agent $\ell_j$ was the first agent in coalition $C_j$ for all $j\in [m]$.
	The algorithm fulfills the property that, for all $j\in [m]$, and agents $x\in C_j\setminus\{\ell_j\}$, it holds that $u_{\ell_j}(x) > 0$.
	We refer to this property as observation $(*)$.
	
	We are ready to prove that $\pi$ is Pareto-optimal.
	Assume that $\pi'$ is any partition such that, for all agents $x\in N$, it holds that $u_x(\pi')\ge u_x(\pi)$.
	We claim that $\pi' = \pi$.
	
	By observation $(*)$ and the design of the algorithm, agent $\ell_1$ is in their best coalition in~$\pi$ and, by the strictness of the utilities, their best coalition is unique, so $\pi'(\ell_1) = \pi(\ell_1) = C_1$.
	We call this fact observation $(**)$.
	We now prove our claim that $\pi' = \pi$ by induction over~$m$, i.e., the number of coalitions in the partition $\pi$.
	
	First, consider the case $m = 1$.
	Then, by observation $(**)$, it holds that $\pi' = \{\pi'(\ell_1)\} = \{\pi(\ell_1)\} = \pi$.

	Now, assume that $m > 1$.
	By observation $(**)$, it suffices to show that $\pi'\setminus\{C_1\} = \pi\setminus\{C_1\}$.
	Consider the ASHG restricted to the agent set $N' = \bigcup_{2\le j\le m} C_j$.
	Then, $\pi \setminus \{C_1\}$ is the output of \Cref{alg:POstrictASHG} of the restricted ASHG if the arrival order is the subsequence of the original arrival order.
	Moreover, by assumption, for all $x\in N$, it holds that $u_x(\pi'\setminus\{C_1\})\ge u_x(\pi\setminus\{C_1\})$.
	Hence, by induction, it holds that $\pi'\setminus\{C_1\} = \pi\setminus\{C_1\}$, as desired.
	This shows that $\pi' = \pi$.
	
	Consequently, no partition exists that Pareto-dominates $\pi$, and therefore $\pi$ is Pareto-optimal.
\end{proof}

Since every Pareto-optimal partition is also a CIS partition, we obtain the following corollary.

\begin{corollary}\label{cor:strictCIS}
	There exists a deterministic online algorithm, which always outputs a CIS partition for strict ASHGs.
\end{corollary}

Of course, \Cref{thm:algPOstrict} and \Cref{cor:strictCIS} work for subclasses of ASHGs like FEGs, AFGs, and AEGs. 
Interestingly, however, the situation changes once we allow for a utility of $0$.
It becomes impossible to compute CIS partitions, and thus PO partitions, online, and this already holds for symmetric FENGs.
This is a clear contrast to offline hedonic games, where PO (and CIS) partitions are guaranteed to exist without any restriction on the game.

\begin{proposition}\label{prop:noCISsymFENG}
	There exists no deterministic online algorithm, which always outputs a CIS partition for symmetric FENGs.
\end{proposition}

\begin{figure}
	\centering
	\begin{tikzpicture}
		
		\node[draw, circle](a1) at (150:1.2){\color{white}{$v3$}};
		\node at (a1){$a$};
		\node[draw, circle](b1) at (30:1.2){\color{white}{$v3$}};
		\node at (b1) {$b$};
		\node[draw, circle](c1) at (270:1.2){\color{white}{$v3$}};
		\node at (c1) {$c$};
		
		\draw (a1) edge node[pos = 0.5, fill = white] {$0$} (b1);
		\draw (a1) edge node[pos = 0.5, fill = white]{$-1$}  (c1);
		\draw (b1) edge node[pos = 0.5, fill = white]{$1$}  (c1);

		\node[draw, circle](a2) at ($(5,0)+(150:1.2) $){\color{white}{$v3$}};
		\node at (a2){$a$};        
		\node[draw, circle](b2) at ($(5,0)+(30:1.2)$){\color{white}{$v3$}};
		\node at (b2) {$b$};       
		\node[draw, circle](c2) at ($(5,0)+(270:1.2)$){\color{white}{$v3$}};
		\node at (c2) {$c$};
		
		\draw (a2) edge node[pos = 0.5, fill = white]{$0$} (b2);
		\draw (a2) edge node[pos = 0.5, fill = white]{$1$} (c2);
		\draw (b2) edge node[pos = 0.5, fill = white]{$1$} (c2);

		\node[draw, circle](a) at ($(0,2) + (150:1.2) + (barycentric cs:a1=1,b2=1)$){\color{white}{$v3$}};
		\node at (a){$a$};
		\node[draw, circle](b) at ($(0,2) + (30:1.2) + (barycentric cs:a1=1,b2=1)$){\color{white}{$v3$}};
		\node at (b){$b$};
		
		\draw (a) edge node[pos = 0.5, fill = white]{$0$} (b);

		\draw[bend right] ($(a)+(-.5,-.3)$) edge[->] ($(0,.5) + (barycentric cs:a1=1,b1=1)$);
		\draw[bend left] ($(b)+(.5,-.3)$) edge[->] ($(0,.5) + (barycentric cs:a2=1,b2=1)$);
		
	\end{tikzpicture}
	\caption{Adversarial FENGs for computing CIS partitions.
	Every deterministic algorithm fails for one of the two possible instances.}
	\label{fig:noCISalg}
\end{figure}

\begin{proof}
	Assume for contradiction that $\alg$ always outputs a CIS partition for symmetric FENGs.
	Consider the following two AFGs $(N,u_1)$ and $(N,u_2)$ with identical agent set $N = \{a,b,c\}$ and symmetric utilities $u_1(a,b) = u_2(a,b) = 0$, $u_1(a,c) = -1$, $u_1(b,c) = 1$, and $u_2(a,c) = u_2(b,c) = 1$.
	Consider the arrival order $a$, then $b$, then~$c$.
	Before the arrival of $c$, $\alg$ cannot distinguish, whether the game will be $(N,u_1)$ or $(N,u_2)$.
	\Cref{fig:noCISalg} depicts the situation.

	If, at the arrival of $b$, $\alg$ forms $\{a,b\}$, assume that we are in game $(N,u_1)$. 
	Then, $\{\{a,b\},\{c\}\}$ is not a CIS partition, because $b$ has a CIS deviation to join agent $c$.
	However, forming $\{a,b,c\}$ does not lead to a CIS partition either, because then agent $a$ has a CIS deviation to form a singleton coalition.
	
	If, however, $\alg$ forms two singleton coalitions for $a$ and $b$, then consider $(N,u_2)$. 
	If $\alg$ forms $\{a,c\}$ or $\{c\}$, then $b$ has a CIS deviation to join this coalition. 
	Finally, if $\alg$ forms $\{b,c\}$, then $a$ has a CIS deviation to join.
\end{proof}

Similar to the previous section, we can extend this result to randomized algorithms.

\begin{restatable}{theorem}{CISImpossible}\label{thm:CISImpossible}
	Let $\alg$ be any randomized online algorithm for symmetric FENGs.
	Then, it holds that $W_{\textnormal{CIS}}(\alg) = 0$.
\end{restatable}

\begin{proof}
	Let $k \in \mathbb{N}$ be a positive integer.
	We define the random FENG $G = (N, u)$ where $N = \bigcup_{i\in [k]}N_i$ for $N_i = \{a_i,b_i,c_i\}$.
	The random utilities are given by $u(a_i, b_i) = 0$ and $u(a_i, c_i) = 1$, and we randomize uniformly between $u(b_i, c_i) = 1$ and $u(b_i, c_i) = -1$.
	All other utilities are set to~$0$.
	The randomization for the utilities within $N_i$ and $N_j$ for $1\le i < j\le k$ are performed independently.
	The agents arrive in the order $N_1, \dots, N_k$ and within $N_i$, first $a_i$, then~$b_i$, then $c_i$.
	Hence, for $i \in [k]$, $G[N_i]$ is one of the two FENGs from \Cref{prop:noCISsymFENG} selected by an unbiased coin flip.
	
	Now, let $\alg$ be an arbitrary deterministic algorithm for symmetric FENGs and define $\pi:= \alg(G)$. 
	By the proof of \Cref{prop:noCISsymFENG}, for every $i\in [k]$, $\alg$ fails with probability $\nicefrac{1}{2}$ on $G[N_i]$.
	
	Moreover, by design of the random instance, if $\alg$ computes a CIS partition on $G$, then, for all $i\in [k]$, $\pi[N_i]$ is CIS for $G[N_i]$.
	Indeed,	assume that $\pi$ is CIS but there exists $i\in[k]$ such that $\pi[N_i]$ is not CIS for $G[N_i]$.
	Then, some agent $x \in N_i$ has a CIS deviation in $G^k[N_i]$ with respect to $\pi[N_i]$.
	However, since the utilities of $x$ to all agents in $N \setminus N_i$ are $0$, they permit $x$ to leave or join their respective coalitions and do not influence the utility change of~$x$.
	Therefore, $\pi$ is not CIS, a contradiction.

	Now, observe that, by the arrival sequence of the agents, the performance of $\alg$ on $N_i$ is at most as good as the performance of the best algorithm for $G[N_i]$.
	Therefore, using the independence of the random selection of the utilities, the probability that $\alg$ computes a CIS partition on $G$ is bounded by the product of the probabilities that $\pi[N_i]$ is CIS for $G[N_i]$.
	Hence, $\pi$ is a CIS partition with probability at most $2^{-k}$.
	
	By Yao's principle, no randomized algorithm can compute a CIS partition with probability more than $2^{-k}$ for every (deterministic) symmetric FENG.
	Since $k$ is chosen arbitrarily, this proves the assertion.
\end{proof}

\subsection{Individual Stability}
As a last single-deviation stability concept, we consider individual stability, which is a strengthening of contractual individual stability and the complementary (but logically incomparable) notion of contractual Nash stability, where each agent in the welcoming (instead of abandoned) partition has the power to veto a single-agent deviation.
Then, even for the combination of symmetry and severe utility restrictions, online algorithms fail to be able to compute IS partitions.
Note that reasonable classes of games contain at least one negative and one positive utility value.
Otherwise, the partition consisting of the grand coalition containing all agents or the partition consisting of singleton coalitions is stable.
By contrast, we show next that computing IS partitions becomes difficult when any positive and any negative utility is present.

Compared to the proofs of \Cref{thm:CNSImpossible,thm:CISImpossible}, simply concatenating identical games by negative utilities can be problematic for some utility values. 
For instance, if the positive utility is sufficiently large compared to the negative utility (e.g., in AFGs), then the grand coalition is IS if each agent has a positive utility for some other agent. 
Instead, we prove the statement by considering one large random adversarial instance for deterministic algorithms and apply Yao's principle once again.

\begin{restatable}{theorem}{noISgen}\label{thm:noISgen}
	Let $x, y > 0$ and let $\alg$ be any randomized online algorithm for symmetric $\{-y,x\}$-ASHGs, symmetric AFGs, or symmetric AEGs.
	Then, $W_{\textnormal{IS}}(\alg) = 0$.
\end{restatable}

\begin{proof}
	Let $x, y > 0$.
	We will define a random adversarial symmetric $\{-y,x\}$-ASHG based on an integer $k \ge 2$ with $n = k^2 + 2$ agents.
	We will then show that computes an IS partition in this game with probability at most $\frac 1k$.
	This result holds independent of the specific values for $x$ and $y$ and we can therefore assume that these values depend on $k$ (and therefore $n$).
	Hence, the construction works in particular for AFGs and AEGs.
	
	We define the game $G = (N,u)$, which is illustrated in \Cref{fig:noISgen}.
	Let $A = \{a_i \colon 1 \le i \le k^2\}$ and $N = A\cup \{b, c\}$.
	Agents arrive in the order $a_1, \dots, a_{k^2}$, then $b$, then $c$.
	Independent of randomizations, it always holds that $u(b, c) = x$ and $u(a_i, a_j) = -y$ for all $a_i, a_j \in A$.
	The remaining utilities are selected at random as follows.
	First, we uniformly draw a random subset $B\subseteq A$ of $k$ agents. 
	Second, one agent is selected from~$B$ uniformly at random and labeled $d$.
	We set $u(c, d) = x$, $u(c, a_i) = -y$ for all $a_i \in A \setminus \{d\}$, $u(b, a_i) = x$ for all $a_i \in B$, and $u(b, a_i) = -y$ for all $a_i \in A \setminus B$.
	Therefore, $G$ is a uniformly random choice from a set of ${{k^2}\choose{k}}k$ instances.

	\begin{figure}
	\centering
	\begin{tikzpicture}
		\node[draw, circle](b) at (1.5,-1){\color{white}{$v3$}};
		\node at (1.5,-1){$b$};
		\node[draw, circle](c) at (4.5,-1){\color{white}{$v3$}};
		\node at (4.5,-1) {$c$};
		
		\foreach[count = \j] \i/\k in {0/{a_1},1/{a_2},3/{d},4.6/{a_{k^2}}}
		{
			\node[draw, circle] (a\j) at (1.5*\i,-2.5){\color{white}{$v3$}};
			\node at (1.5*\i,-2.5){$\k$};
		}
		\node at (barycentric cs:a3=1,a4=1) {\dots};
		\node at (barycentric cs:a2=1,a3=1) {\dots};
		
		\draw (b) edge node[pos = 0.5, fill = white]{$x$} (a1);
		\draw (b) edge node[pos = 0.5, fill = white]{$x$} (a2);
		\draw (b) edge node[pos = 0.5, fill = white]{$x$} (a3);
		\draw (b) edge node[pos = 0.5, fill = white]{$x$} (c);
		\draw (c) edge node[pos = 0.5, fill = white]{$x$} (a3);
		
		\draw[decorate,decoration={brace,amplitude=12pt}] ($(a3)+(.5,-.3)$) -- ($(a1)+(-.5,-.3)$) node[midway, below,yshift=-12pt,]{$B$};
	\end{tikzpicture}
	\caption{Adversarial instance for achieving individual stability in $\{-y,x\}$-ASHGs for $x,y>0$.
		We only depict the positive utilities of $x$. All remaining utilities are $-y$.}
	\label{fig:noISgen}
\end{figure}

	Before we bound the probability of a deterministic algorithm for forming an IS partition, we determine the IS partitions in the obtained instance dependent on $x$ and $y$.
	
	\begin{restatable}{claim}{ISclaim}\label{cl:ISparitions}
		Let $\emptyset \neq S \subseteq B \setminus \{d\}$.
		The following are individually stable partitions in $G$:
		\begin{itemize}
			\item $\{\{b, c, d\}\} \cup \{\{a_i\}\colon a_i \in A \setminus \{d\}\}$ for all $x,y > 0$,
			\item $\{\{c, d\}, \{b\} \cup S\} \cup \{\{a_i\}\colon a_i \in A \setminus \left(\{d\} \cup S\right) \}$ if $2 \le |S| \le \frac xy +1$, and 
			\item $\{\{b, c, d\} \cup S\} \cup \{\{a_i\}\colon a_i \in A \setminus \left(\{d\} \cup S\right) \}$ if $|S| \le \frac xy -1$ 
		\end{itemize}
		Moreover, no other partition is individually stable.
	\end{restatable}

	\renewcommand\qedsymbol{$\vartriangleleft$} 
	\begin{proof}
		Let $\pi$ be an IS partition for $G$.
		Independently of $x$ and $y$, it has to hold that all agents in $A \setminus B$ are in singleton coalitions, and all agents in $B$ are either in a coalition with $b$, or in a coalition with $c$, or they are in a singleton coalition as well.
		We disregard the agents in $A \setminus B$ from consideration for the rest of the proof.
		Another important observation is that whenever a pair of agents has positive utility, then agent~$b$ is involved in all but one case.
		We base the proof on a case distinction depending on $\pi(b)$.
		
		First, it holds that $\pi(b) \neq \{b\}$, as otherwise every agent in $B \setminus \{d\}$ has an IS deviation to join~$\{b\}$, a contradiction.
		Second, it holds that $\pi(b) \neq \{b, c\}$ and $\pi(b) \neq \{b, d\}$, as otherwise agent~$d$ or $c$ respectively have an IS deviation to join $\pi(b)$.
		
		Now, assume that there exists $a_i \in B \setminus \{d\}$ such that $\pi(b) = \{a_i, b\}$.
		Then, $\{c, d\}\notin \pi$, as otherwise agent $b$ has an IS deviation to join $\{c, d\}$.
		Hence, $c$ and $d$ must be in singleton coalitions as they have IS deviations to leave any other coalition to form a singleton coalition.
		However, then, they both have the IS deviation to join each other, which is a contradiction. Hence, $\pi(b) \neq \{a_i, b\}$.
		
		Next, assume that there exists $\emptyset \neq S \subseteq B \setminus \{d\}$ such that $\pi(b) = S \cup \{b, c\}$.
		Then, agent~$c$ has an IS deviation to join $\{d\}$ if this coalition exists.
		Otherwise, $d$ has an IS deviation to leave their coalition and form a singleton coalition.
		Hence, this case is also excluded.
		
		Now, assume that there exists $\emptyset \neq S \subseteq B \setminus \{d\}$ such that $\pi(b) = S \cup \{b, d\}$.
		Then, agent $d$ has an IS deviation to join $\{c\}$ if this coalition exists.
		Otherwise, $c$ has an IS deviation to leave their coalition and form a singleton coalition.
		
		So far, we have excluded several cases in which no IS partition is possible.
		We conclude that $b$ must be in a coalition of size at least $3$ and that the coalition of $b$ must contain either both $c$ and $d$ or none of them.
		In the remaining cases, we find some IS partitions.
		
		First, consider the case where $\pi(b) = \{b, c, d\}$.
		Then $\pi(a_i) = \{a_i\}$ for all $a_i\in A$, and the resulting partition is an IS partition.
		This proves that the first partition of the claim is an IS partition.
		
		Next, assume that there exists $\emptyset \neq S \subseteq B \setminus \{d\}$ such that $\pi(b) = S \cup \{b\}$.
		Then, all agents $a_i \in A \setminus \left(S \cup \{d\}\right)$ must be in singleton coalitions, as otherwise they have an IS deviation to form a singleton coalition.
		This only leaves agents $c$ and $d$, and we conclude that $\{c,d\}\in \pi$, as otherwise, they have an IS deviation to join each other.
		Finally, the partition is only IS if $|S|\le \frac xy+1$. 
		Otherwise, 
		all agents $a_i \in S$ have a utility of $x - (|S| - 1) y < 0$ and thus can perform an IS deviation to form a singleton coalition.
		In addition, we need $|S|\ge 2$ as otherwise $b$ performs a deviation to join $\{c,d\}$.
		
		Moreover, for $2\le |S|\le \frac xy+1$, the partition $\{\{c, d\}, \{b\} \cup S\} \cup \{\{a_i\}\colon a_i \in A \setminus \left(\{d\} \cup S\right) \}$ can be shown to be individually stable:
		Clearly, none of the agents in $A$ in singleton coalitions can enter $\pi(b)$ because this would be blocked by agents in $S$.
		Similarly, they are also blocked to enter other singleton coalitions.
		Agents $b$ and $c$ are also blocked to join any other coalition.
		Next, agents in $S$ cannot improve by performing any deviation.
		Finally, because $|S|\ge 2$, $b$ cannot improve by joining any other coalition.
		
		Together, for $\emptyset \neq S \subseteq B \setminus \{d\}$, the partition $\{\{c, d\}, \{b\} \cup S\} \cup \{\{a_i\}\colon a_i \in A \setminus \left(\{d\} \cup S\right) \}$ is an IS partition if and only if $2 \le |S| \le \frac xy +1$.
		
		Finally, assume that there exists  $\emptyset \neq S \subseteq B \setminus \{d\}$ such that $\pi(b) = S \cup \{b, c, d\}$.
		Then, as before, all agents $a_i \in A \setminus \left(S \cup \{d\}\right)$ must be in singleton coalitions.
		Moreover, the partition only is individually stable if $|S| \le \frac xy -1$.
		Otherwise, all agents $a_i \in S$ have a utility of $x - (|S|+1) \cdot y < 0$ and can thus perform an IS deviation to form a singleton coalition.
		Agents $b$, $c$, and $d$ all have two friends and thus no IS deviation when $|S| \le \frac xy -1$.
		This proves that  for $\emptyset \neq S \subseteq B \setminus \{d\}$, the partition $\{\{b, c, d\} \cup S\} \cup \{\{a_i\}\colon a_i \in A \setminus \left(\{d\} \cup S\right) \}$ is individually stable if and only if $|S| \le \frac xy -1$.
		
		As this case distinction covers all possible partitions, we have found all IS partitions as stated in the assertion.
	\end{proof}
	\renewcommand\qedsymbol{$\square$}

	Now, let $\alg$ be any deterministic online algorithm and define $\pi := \alg(G)$.
	To conclude the proof, we show that $\pi$ is an IS partition with probability at most $\frac{1}{k}$.
	
	Note that $G[A]$ is identical independent of the randomization for the instance.
	We perform a case distinction based on $\pi[A]$.
	Intuitively, $\alg$ can either attempt to reach the IS partition $\{\{b, c, d\}\} \cup \{\{a_i\}\colon a_i \in A \setminus \{d\}\}$ by forming all singleton coalitions, or it forms a single coalition of size greater than one to reach any of the other IS partitions.
	In all other cases, the algorithm can no longer reach an IS partition.

	Assume first that $\pi[A]$ contains exactly one coalition of size strictly larger than one.
	Note that then $\pi\neq \{\{b, c, d\}\} \cup \{\{a_i\}\colon a_i \in A \setminus \{d\}\}$ and it can only create an IS partition for the latter two cases of \Cref{cl:ISparitions}.
	Let $S \subseteq A$ be the coalition of size strictly larger than $1$ in $\pi[A]$.
	Then, $\pi$ can only be individually stable if $S \subseteq B$ for the random set $B$, particularly $|S| \le k$.
	As $B$ is chosen uniformly at random from $A$, the probability that $S \subseteq B$ is $\frac{\binom{k}{|S|}}{\binom{k^2}{|S|}}$.
	This is true because there are $\binom{k^2}{|S|}$ choices to select $S$ and $\binom{k}{|S|}$ of these choices result in $S\subseteq B$.
	We compute
	\begin{align*}
		\frac{\binom{k}{|S|}}{\binom{k^2}{|S|}} = & \frac{k!|S|!(k^2-|S|)!}{k^2!|S|!(k - |S|)!} 
		= \frac{k (k - 1) \cdots (k - |S| + 1)}{k^2 (k^2 - 1) \cdots (k^2 - |S| + 1)} \\
		= & \frac{1}{k}\frac{k-1}{k^2-1}\cdots \frac{k - |S| + 1}{k^2 - |S| + 1} < \frac{1}{k}\text.
	\end{align*}
	
	Therefore, $\alg$ successfully forms an IS partition with probability at most $\frac{1}{k}$ in this case.
	
	Next, assume that $\pi[A]$ contains only singleton coalitions.
	By \Cref{cl:ISparitions}, $\alg$ only outputs an IS partition if $\pi = \{\{b, c, d\}\} \cup \{\{a_i\}\colon a_i \in A \setminus \{d\}\}$.
	Therefore, when $b$ arrives, the set $B$ is revealed to the algorithm, and it needs to match $b$ to $d$, as otherwise $\{b,c,d\}$ cannot be formed.
	The probability for this event is precisely $\frac{1}{k}$ since each element in $B$ is $d$ with equal probability.
	Hence, also in this case, $\pi$ is IS with probability at most $\frac{1}{k}$.
	
	Together, $\pi$ is an IS partition with probability at most $\frac 1k$.
	By Yao's principle, no randomized algorithm can compute an IS partition with probability more than $\frac 1k$ for every (deterministic) symmetric $\{-y,x\}$-ASHGs.
	Since our choice of $k$ is arbitrary, the assertion follows.
\end{proof}

\subsection{Group Stability}

Finally, we consider group stability, that is, computing partitions in the core and strict core.
Since partitions in the strict core are individually stable, \Cref{thm:noISgen} already suggests computational difficulties for achieving partitions in the strict core.
There is, however, a caveat.
For some parameters of $x$ and $y$, the instances considered in \Cref{thm:noISgen} have an empty strict core.
Hence, an online algorithm for these instances is not worse than the best offline algorithm (see also \Cref{fn:onlinevsoffline}).
In this section, we therefore restrict attention to instances containing partitions in the strict core.
Note that this assumption is, for instance, trivially true for AFGs, in which partitions in the strict core are guaranteed to exist \citep{DBHS06a}. 

We now prove a negative result encompassing both the core and the strict core.
The proof is similar to the proof of \Cref{thm:noISgen}, but a suitable set of agents arrives instead of the single agent $c$.
In addition, we can simplify the instances from \Cref{thm:noISgen} because we can omit agents in $A\setminus B$.

\begin{theorem}\label{thm:core}
	Let $x, y > 0$ and $\alg$ be any randomized online algorithm for symmetric $\{-y,x\}$-ASHGs (or symmetric AEGs) that contain partitions in the strict core.
	Then, it holds that $W_{\textnormal{CR}}(\alg) = 0$ and $W_{\textnormal{SCR}}(\alg) = 0$.
\end{theorem}

\begin{proof}
	Let $x, y > 0$ and define $q = \min\left\{q'\in \mathbb N\colon q'> \frac xy\right\}$.
	Note that $q$ is a fixed parameter of our construction that only depends on~$x$ and~$y$.
	
	We will define a random adversarial symmetric $\{-y,x\}$-ASHG $G = (N,u)$ based on $q$ as well as an integer $k \ge 2$ with $n = k + q + 1$ agents.
	\Cref{fig:noCORE} illustrates the game.
	Let $A = \{a_i \colon 1 \le i \le k\}$, $C = \{c_i\colon 1\le i\le q\}$ and $N = A\cup C\cup \{b\}$, where $b$ is an additional agent.
	Agents arrive in the order $a_1, \dots, a_{k}$, then $b$, then $c_1,\dots, c_q$.
	We define
	\begin{itemize}
		
		\item for all $i\in [q]$, $u(b, c_i) = x$,
		\item for all $i\in [k]$, $u(b, a_i) = x$,
		\item for all $i,j\in [q]$ with $i\neq j$, $u(c_i, c_j) = x$, and
		\item for all $i,j\in [k]$ with $i\neq j$, $u(a_i, a_j) = -y$.
	\end{itemize}
	
	The remaining utilities are selected at random as follows.
	We uniformly draw an index $j^*\in [k]$.
	Then, for all $i\in [q]$ and $j\in [k]$ with $j\neq j^*$, we set $u(c_i, a_{j^*}) = x$ and $u(c_i, a_j) = -y$, i.e., agent $a_{j^*}$ has positive utility for all agents in $C$ while all other agents in $A$ have negative utilities.
	Therefore, $G$ is a uniformly random choice from a set of $k$ instances.

	\begin{figure}
	\centering
	\begin{tikzpicture}
		\node[draw, circle](b) at (1.5,-1){\color{white}{$v3$}};
		\node at (1.5,-1){$b$};
		\node[draw, circle](c1) at (4.5,-1){\color{white}{$v3$}};
		\node at (4.5,-1) {$c_1$};
		\node[draw, circle](c2) at (6.5,-1){\color{white}{$v3$}};
		\node at (6.5,-1) {$c_q$};
		
		\foreach[count = \j] \i/\k in {0/{a_1},1/{a_2},3/{a_{j^*}}}
		{
			\node[draw, circle] (a\j) at (1.5*\i,-2.5){\color{white}{$v3$}};
			\node at (1.5*\i,-2.5){$\k$};
		}
		\node at (barycentric cs:c1=1,c2=1) {\dots};
		\node at (barycentric cs:a2=1,a3=1) {\dots};
		
		\draw (b) edge node[pos = 0.5, fill = white]{$x$} (a1);
		\draw (b) edge node[pos = 0.5, fill = white]{$x$} (a2);
		\draw (b) edge node[pos = 0.5, fill = white]{$x$} (a3);
		\draw (b) edge node[pos = 0.5, fill = white]{$x$} (c1);
		\draw (c1) edge node[pos = 0.5, fill = white]{$x$} (a3);
		\draw (c2) edge node[pos = 0.5, fill = white]{$x$} (a3);
		\draw[bend left] (c1) edge node[pos = 0.3, fill = white]{$x$} (c2);
		\draw[bend left] (b) edge node[pos = 0.5, fill = white]{$x$} (c2);
		
	\end{tikzpicture}
	\caption{Adversarial instance for achieving partitions in the (strict) core in $\{-y,x\}$-ASHGs for $x,y>0$.
	All such instances have a nonempty strict core.
		We only depict the positive utilities of~$x$. All remaining utilities are~$-y$.}
	\label{fig:noCORE}
\end{figure}

First, we ensure that $G$ always contains partitions in the strict core.
Define the partition $\pi^* = \{C\cup \{a_{j^*}, b\}\} \cup \{\{a_j\}\colon 1\le j \le k, j\neq j^*\}$.

\begin{claim}\label{cl:StrCore}
	The partition $\pi^*$ is in the strict core.
\end{claim}
	\renewcommand\qedsymbol{$\vartriangleleft$}
	\begin{proof}
		First, note that an agent that is part of their unique best coalition can never be part of a weakly blocking coalition.
		Hence, since this is the case for agents in $C\cup \{a_{j^*}\}$, we only have to exclude weakly blocking coalitions containing $b$ and agents in $A\setminus \{a_{j^*}\}$.
		Let $D\subseteq (A\setminus \{a_{j^*}\})\cup \{b\}$ with $b\in D$ and assume that $u_b(D) \ge u_b(\pi^*)$.
		Then, $|D\cap (A\setminus \{a_{j^*}\})|\ge q +1$, otherwise coalition $D$ lowers agent $b$'s utility compared to $\pi^*(b)$.
		Hence, for $d\in D\setminus \{b\}$, it holds that $u_d(D) = x - (|D|-2)y < x - q y < 0$, where we use that $q > \frac xy$.
		It follows that $D$ is not a weakly blocking coalition, and therefore, $b$ is not part of a weakly blocking coalition.
		However, the agents in $A\setminus \{a_{j^*}\}$ all have a negative utility for each other, and therefore, they cannot form a weakly blocking coalition.
		Hence, there is no weakly blocking coalition, and $\pi^*$ is in the strict core.
	\end{proof}
	
The previous claim shows that our considered instances have a suitable form, i.e., they contain elements in the strict core.
To continue the proof, we show that $\pi^*$ is the only partition in the strict core and even in the core.

\begin{claim}
	The partition $\pi^*$ is the unique partition in the core.
\end{claim}

\begin{proof}
	By \Cref{cl:StrCore}, we already know that $\pi^*$ is in the strict core and, therefore, in the core.
	It remains to show that the core does not contain any other partition.
	Let $\pi$ be a partition in the core.
	We will show that $\pi = \pi^*$.
	
	First, we will prove that there exists a coalition $D\in \pi$ with $C\cup \{a_{j^*}\}\subseteq D$.
	To prove this, observe that the coalition $C\cup \{a_{j^*}\}$ yields a utility of $qx$ to all its members.
	Hence, because $\pi$ does not admit a blocking coalition, there exists an agent $d\in C\cup \{a_{j^*}\}$ with $u_d(\pi)\ge qx$.
	This can only happen if the coalition of~$d$ contains at least $q$ agents for which they receive a positive utility.
	Hence if $b\notin \pi(d)$, then $C\cup \{a_{j^*}\}\subseteq \pi(d)$ and our assertion is true.
	Moreover, if there exists an agent $a\in A\setminus \{a_{j^*}\}$ with $a\in \pi(d)$, then $u_d(\pi)\ge qx$ is only possible if all $q+1$ agents for which $d$ achieves a positive utility are in $\pi$.
	Then, once again $C\cup \{a_{j^*}\}\subseteq \pi(d)$.
	
	Together, it remains to consider the case where $b\in \pi(d)$ and $\pi(d) \subseteq C\cup \{a_{j^*},b\}$.
	We are done if $\pi(d) = C\cup \{a_{j^*},b\}$.
	Otherwise, because $|\pi(d)|\ge q + 1$, there is a unique agent $d'\in (C\cup \{a_{j^*}\})\setminus \pi(d)$.
	However, forming $C\cup \{a_{j^*},b\}$ is preferred by all agents in $\pi(d)$ as well as by $d'$.
	Hence, this is a blocking coalition, contradicting the fact that $\pi$ is in the core.
	Thus, it must hold that $\pi(d) = C\cup \{a_{j^*},b\}$.
	
	Now, consider the coalition $D\in \pi$ with $C\cup \{a_{j^*}\}\subseteq D$.
	Then, for all $1\le j\le k$ with $j\neq j^*$, it holds that $a_j\notin D$.
	Otherwise, $u_{a_j}(\pi) \le x - (q+1)y < 0$ and $\{a_j\}$ would be a blocking coalition.
	
	Next, assume for contradiction that $b\notin D$.
	Note that all members in $D$ prefer $D\cup \{b\}$.
	Hence, for this not to be a blocking coalition, it must hold that $u_b(\pi)\ge x(q+1)$.
	Therefore, $|\pi(b) \cap (A\setminus \{a_{j^*}\})|\ge q+1$.
	But then, for $a\in \pi(b)\setminus \{b\}$, it holds that $u_a(\pi) \le x - qy < 0$.
	This is a contradiction.
	Hence, it must hold that $b\in D$.
	Taken together, we conclude that $D = C\cup \{a_{j^*},b\}$.
	
	For the remaining agents, i.e., for agents in $A\setminus \{a_{j^*}\}$, every nonempty coalition among themselves yields a negative utility.
	Hence, these agents have to form singleton coalitions in $\pi$.
	We conclude that $\pi^* = \pi$.
\end{proof}
	\renewcommand\qedsymbol{$\square$}
	
	Now, let $\alg$ be any deterministic online algorithm and define $\pi := \alg(G)$.
	To conclude the proof, we show that $\pi$ is in the core (and therefore in the strict core) with probability at most $\frac{1}{k}$.
	
	Recall that the first~$k$ agents to arrive are the agents in $A$. 
	Since $\pi^*$ is the only coalition in the core, we can restrict attention to the case where $\alg$ assigns all agents in $A$ to singleton coalitions.
	If $\alg$ forms a singleton coalition when $b$ arrives, then $\pi\neq \pi^*$ and $\pi$ is not in the core.
	Assume that $\alg$ forms a coalition of $b$ with an agent from $A$.
	Then, $\pi = \pi^*$ is only possible if $b$ forms a coalition with $a_{j^*}$.
	The probability for this event is exactly $\frac{1}{k}$ since $j^* = j$ for all $j\in [k]$ with equal probability.
	Hence, $\pi$ is in the core with probability at most $\frac{1}{k}$.
	
	By Yao's principle, no randomized algorithm can compute a partition in the core (and therefore in the strict core) with probability more than $\frac 1k$ for every (deterministic) symmetric $\{-y,x\}$-ASHGs containing a partition in the strict core.
	Since our choice of $k$ is arbitrary, the assertion follows.
	
	Finally, we observe that our construction also works for AEGs.
	These are $\{-y,x\}$-ASHGs, where $x = 1$ and $y$ depends on the number of agents.
	We can then simply set $q = 1$. 
	Since all our games contain at least $2$ agents, we then have that $q > \frac 1n$, i.e., our construction is valid for $x = 1$ and $y = n$.
\end{proof}

Note that the construction in the previous proof does not work for AFGs because of the dependence of $q$ on $x$: we cannot simultaneously satisfy $q > \frac n1$ and $n = k + q + 1$ when $k\ge 2$.
Instead, we provide a different construction to obtain a result for AFGs analogous to \Cref{thm:core}.

\begin{theorem}\label{thm:coreAFG}
	Let $\alg$ be any randomized online algorithm for symmetric AFGs.
	Then, it holds that $W_{\textnormal{CR}}(\alg) = 0$ and $W_{\textnormal{SCR}}(\alg) = 0$.
\end{theorem}

\begin{proof}
		Let $k\ge 3$. 
		We will define a random adversarial symmetric AFG $G = (N,u)$ with $n = k + 1$ agents.
		\Cref{fig:noCOREafg} illustrates the game.
		Let $A = \{a_i \colon 1 \le i \le k\}$ and define $N = A\cup \{b\}$, where $b$ is an additional agent.
		Agents arrive in the order $a_1, \dots, a_{k}$, then $b$.
		For all $i,j\in [k]$ with $i\neq j$, we define $u(a_i, a_j) = -1$.
	
		The remaining utilities are selected at random as follows.
		We uniformly draw an index set $J\subseteq [k]$ with $|J| = 3$, say $J = \{j_1,j_2,j_3\}$.
		We set $u(b, a_j) = n$ if $j\in J$ and $u(b, a_j) = -1$ if $j\in [k]\setminus J$.
		Hence, there are exactly three positive utilities from $b$ to a random subset of agents in $A$ and all other utilities are negative.
		Therefore, $G$ is a uniformly random choice from a set of $\binom{k}{3}$ instances.

		\begin{figure}
		\centering
		\begin{tikzpicture}
			\node[draw, circle](b) at (7.5,-1){\color{white}{$v3$}};
			\node at (7.5,-1){$b$};
		
			\foreach[count = \j] \i/\k in {1/{a_1},2/{a_2},3.5/{a_{j_1}},5/{a_{j_2}},6.5/{a_{j_3}},8/{a_{k}}}
			{
				\node[draw, circle] (a\j) at (1.5*\i,-2.5){\color{white}{$v3$}};
				\node at (1.5*\i,-2.5){$\k$};
			}
			
			\node at (barycentric cs:a2=1,a3=1) {\dots};
			\node at (barycentric cs:a3=1,a4=1) {\dots};
			\node at (barycentric cs:a4=1,a5=1) {\dots};
			\node at (barycentric cs:a5=1,a6=1) {\dots};
		
			\draw (b) edge node[pos = 0.5, fill = white]{$n$} (a3);
			\draw (b) edge node[pos = 0.5, fill = white]{$n$} (a4);
			\draw (b) edge node[pos = 0.5, fill = white]{$n$} (a5);
		
		\end{tikzpicture}
		\caption{Adversarial instance for achieving partitions in the (strict) core in AFGs.
			We only depict the positive utilities of~$n$. All remaining utilities are~$-1$.}
		\label{fig:noCOREafg}
	\end{figure}

	We start by determining the partitions in the core in the strict core.
	
	\begin{claim}\label{cl:StrCoreAFG}
		The following statements are true.
		\begin{enumerate}
			\item The partition $\{\{b\}\cup \{a_j\colon j\in J\}\}\cup \{\{a_j\}\colon j\in [k]\setminus J\}$ is in the strict core.
			\item For $J'\subseteq J$ with $|J'| = 2$, the partition $\{\{b\}\cup \{a_j\colon j\in J'\}\}\cup \{\{a_j\}\colon j\in [k]\setminus J'\}$ is in the core.
			\item No other partition is in the core.
		\end{enumerate}
	\end{claim}
		\renewcommand\qedsymbol{$\vartriangleleft$}
		\begin{proof}
			Let $\pi$ be a partition in the core.
			First note that for $j\in [k]\setminus J$, all coalitions except a singleton coalition yield a negative utility for $a_j$.
			Hence, $\{a_j\}\in \pi$.
			
			Now, let $C = \pi(b)$ and define $B := \{a_j\colon j\in J\}$.
			We already know that $C\subseteq \{b\}\cup B$.
			If $B = \{b\}$, then $\{b\}\cup B$ is a blocking coalition.
			Next, assume that $|C\cap B| = 1$, say $C = \{b,a_{j_1}\}$.
			Then, $\{b, a_{j_2}, a_{j_3}\}$ is a blocking coalition.
			Hence $\pi$ must be of the form described in the first or second case of the claim, which proves the third assertion.
			
			If $C = \{b\}\cup B$, then $b$ is in their unique most preferred coalition, and therefore not part of a weakly blocking coalition. 
			However, no coalition consisting only of agents in $A$ can be a weakly blocking coalition.
			Hence, $\pi$ is in the strict core, proving the first assertion.
			
			Finally, if $|C\cap B| = 2$, then the only coalition that is better for $b$ is $\{b\}\cup B$, which is worse for the other agents in $C$.
			Similar as before, no coalition consisting only of agents in $A$ is a (weakly) blocking coalition. 
			Hence, $\pi$ is in the core, proving the second assertion.
		\end{proof}
	\renewcommand\qedsymbol{$\square$}
	
	Now, let $\alg$ be any deterministic online algorithm for symmetric AFGs and define $\pi := \alg(G)$.
	We claim that $\pi$ is in the core (and therefore in the strict core) with probability at most $\frac{6}{k(k-1)}$.
	
	Recall that the first~$k$ agents to arrive are the agents in $A$. 
	By \Cref{cl:StrCoreAFG}, $\pi$ can only be in the core if $\alg$ forms a coalition~$C$ of size~$2$ or size~$3$ among the agents in $A$.
	Even more, $\pi$ can only be in the core if $C\subseteq B$.
	If $|C| = 2$, then $C\subseteq B$ occurs with probability $\frac{\binom{3}{2}}{\binom{k}{2}} = \frac 6{k(k-1)}$. 
	If $|C| = 3$, then $C\subseteq B$ occurs with probability $\frac{1}{\binom{k}{3}} = \frac 6{k(k-1)(k-2)}\le \frac 6{k(k-1)}$. 
	Hence, $\pi$ is in the core (and therefore in the strict core) with probability at most $\frac 6{k(k-1)}$.
	
	By Yao's principle, no randomized algorithm can compute a partition in the core (and therefore in the strict core) with probability more than $\frac 6{k(k-1)}$ for every (deterministic) symmetric AFG.
	Since our choice of $k$ is arbitrary, the assertion follows.
\end{proof}

\section{Conclusion}

In this paper, we have studied stability in online coalition formation.
We have considered stability notions based on single-agent deviations, group deviations, and Pareto optimality.
\Cref{tab:overview} in \Cref{sec:intro} displays an overview of our results.

Our positive results follow from two deterministic algorithms.
The first one outputs CNS partitions for symmetric games with utility restrictions that include FEGs and AEGs, and the second one applies to strict ASHGs and outputs PO partitions.
The latter is interesting because Pareto optimality has the flavor of both stability and optimality.\footnote{Notably, however, both algorithms may output partitions of negative social welfare for instances in which the maximum social welfare is positive, and therefore yield no approximation guarantee for maximizing social welfare.}

By contrast, we obtain negative results in the sense that no randomized algorithm can guarantee any fixed probability to output a stable partition.
Surprisingly, our negative results even encompass concepts like contractual individual stability and Pareto optimality, for which solutions exist in every hedonic game.
Hence, the online capabilities of algorithms can be severely weaker than strong offline possibilities.
Negative results naturally extend to stronger solution concepts.
For instance, a consequence of our negative results for IS and CNS is that the NS guarantee is~$0$ for all considered game restrictions.

We believe that departing from the mere consideration of welfare maximality in online coalition formation is an important step.
There is plenty of space for future research in this direction.
Possible directions include considering other solution concepts, such as fairness notions.
Moreover, it would be interesting to consider game classes that are different from ASHGs.

\section*{Acknowledgments}

This work was supported by the Deutsche Forschungsgemeinschaft under grants BR 2312/11-2 and BR 2312/12-1, and by the AI Programme of The Alan Turing Institute.
A preliminary version of this article appeared in the Proceedings of the 38th AAAI Conference on Artificial Intelligence (February 2024).
We thank Saar Cohen and the anonymous reviewers from AAAI for their helpful comments.


\vskip 0.2in
\begin{thebibliography}{44}
\providecommand{\natexlab}[1]{#1}
\providecommand{\url}[1]{\texttt{#1}}
\expandafter\ifx\csname urlstyle\endcsname\relax
  \providecommand{\doi}[1]{doi: #1}\else
  \providecommand{\doi}{doi: \begingroup \urlstyle{rm}\Url}\fi

\bibitem[Alcalde and Revilla(2004)]{AlRe04a}
Jos{\'e} Alcalde and Pablo Revilla.
\newblock {Researching with whom? Stability and manipulation}.
\newblock \emph{Journal of Mathematical Economics}, 40\penalty0 (8):\penalty0 869--887, 2004.

\bibitem[Aleksandrov and Walsh(2019)]{AlWa19a}
Martin Aleksandrov and Toby Walsh.
\newblock Strategy-proofness, envy-freeness and {P}areto efficiency in online fair division with additive utilities.
\newblock In \emph{Proceedings of the 16th Pacific Rim International Conference on Artificial Intelligence (PRICAI)}, pages 527--541. Springer, 2019.

\bibitem[Aziz and Savani(2016)]{AzSa15a}
Haris Aziz and Rahul Savani.
\newblock Hedonic games.
\newblock In Felix Brandt, Vincent Conitzer, Ulle Endriss, J{\'e}r{\^o}me Lang, and Ariel~D. Procaccia, editors, \emph{Handbook of Computational Social Choice}, chapter~15. Cambridge University Press, 2016.

\bibitem[Aziz et~al.(2013)Aziz, Brandt, and Seedig]{ABS11c}
Haris Aziz, Felix Brandt, and Hans~Georg Seedig.
\newblock Computing desirable partitions in additively separable hedonic games.
\newblock \emph{Artificial Intelligence}, 195:\penalty0 316--334, 2013.

\bibitem[Aziz et~al.(2019)Aziz, Brandl, Brandt, Harrenstein, Olsen, and Peters]{ABB+17a}
Haris Aziz, Florian Brandl, Felix Brandt, Paul Harrenstein, Martin Olsen, and Dominik Peters.
\newblock Fractional hedonic games.
\newblock \emph{ACM Transactions on Economics and Computation}, 7\penalty0 (2):\penalty0 1--29, 2019.

\bibitem[Banerjee et~al.(2001)Banerjee, Konishi, and S{\"o}nmez]{BKS01a}
Suryapratim Banerjee, Hideo Konishi, and Tayfun S{\"o}nmez.
\newblock Core in a simple coalition formation game.
\newblock \emph{Social Choice and Welfare}, 18:\penalty0 135--153, 2001.

\bibitem[Benad{\`e} and Sahoo(2023)]{BeSa23a}
Gerdus Benad{\`e} and Nachiketa Sahoo.
\newblock Stability, fairness and the pursuit of happiness in recommender systems.
\newblock Technical Report 4241170, SSRN, 2023.

\bibitem[Bil{\`o} et~al.(2018)Bil{\`o}, Fanelli, Flammini, Monaco, and Moscardelli]{BFFMM18a}
Vittorio Bil{\`o}, Angelo Fanelli, Michele Flammini, Gianpiero Monaco, and Luca Moscardelli.
\newblock Nash stable outcomes in fractional hedonic games: Existence, efficiency and computation.
\newblock \emph{Journal of Artificial Intelligence Research}, 62:\penalty0 315--371, 2018.

\bibitem[Bil{\`o} et~al.(2022)Bil{\`o}, Monaco, and Moscardelli]{BMM22a}
Vittorio Bil{\`o}, Gianpiero Monaco, and Luca Moscardelli.
\newblock Hedonic games with fixed-size coalitions.
\newblock In \emph{Proceedings of the 36th AAAI Conference on Artificial Intelligence (AAAI)}, pages 9287--9295, 2022.

\bibitem[Boehmer et~al.(2023)Boehmer, Bullinger, and Kerkmann]{BBK23a}
Niclas Boehmer, Martin Bullinger, and Anna~M. Kerkmann.
\newblock Causes of stability in dynamic coalition formation.
\newblock In \emph{Proceedings of the 37th AAAI Conference on Artificial Intelligence (AAAI)}, pages 5499--5506, 2023.

\bibitem[Bogomolnaia and Jackson(2002)]{BoJa02a}
Anna Bogomolnaia and Matthew~O. Jackson.
\newblock The stability of hedonic coalition structures.
\newblock \emph{Games and Economic Behavior}, 38\penalty0 (2):\penalty0 201--230, 2002.

\bibitem[Brandt et~al.(2022)Brandt, Bullinger, and Tappe]{BBT22a}
Felix Brandt, Martin Bullinger, and Leo Tappe.
\newblock Single-agent dynamics in additively separable hedonic games.
\newblock In \emph{Proceedings of the 36th AAAI Conference on Artificial Intelligence (AAAI)}, pages 4867--4874, 2022.

\bibitem[Brandt et~al.(2023)Brandt, Bullinger, and Wilczynski]{BBW21b}
Felix Brandt, Martin Bullinger, and Ana{\"e}lle Wilczynski.
\newblock Reaching individually stable coalition structures.
\newblock \emph{ACM Transactions on Economics and Computation}, 11\penalty0 (1--2):\penalty0 4:1--65, 2023.

\bibitem[Bullinger(2020)]{Bull19a}
Martin Bullinger.
\newblock Pareto-optimality in cardinal hedonic games.
\newblock In \emph{Proceedings of the 19th International Conference on Autonomous Agents and Multiagent Systems (AAMAS)}, pages 213--221, 2020.

\bibitem[Bullinger(2022)]{Bull22a}
Martin Bullinger.
\newblock Boundaries to single-agent stability in additively separable hedonic games.
\newblock In \emph{Proceedings of the 47th International Symposium on Mathematical Foundations of Computer Science (MFCS)}, pages 26:1--26:15, 2022.

\bibitem[Bullinger and Romen(2023)]{BuRo23a}
Martin Bullinger and Ren{\'e} Romen.
\newblock Online coalition formation under random arrival or coalition dissolution.
\newblock In \emph{Proceedings of the 31st European Symposium on Algorithms (ESA)}, pages 27:1--27:18, 2023.

\bibitem[Bullinger and Suksompong(2023)]{BuSu23a}
Martin Bullinger and Warut Suksompong.
\newblock Topological distance games.
\newblock In \emph{Proceedings of the 37th AAAI Conference on Artificial Intelligence (AAAI)}, pages 5549--5556, 2023.

\bibitem[Carosi et~al.(2019)Carosi, Monaco, and Moscardelli]{CMM19a}
Raffaello Carosi, Gianpiero Monaco, and Luca Moscardelli.
\newblock Local core stability in simple symmetric fractional hedonic games.
\newblock In \emph{Proceedings of the 18th International Conference on Autonomous Agents and Multiagent Systems (AAMAS)}, pages 574--582, 2019.

\bibitem[Cechl{\'a}rov{\'a} and Romero-Medina(2001)]{CeRo01a}
Katar{\'\i}na Cechl{\'a}rov{\'a} and Antonio Romero-Medina.
\newblock Stability in coalition formation games.
\newblock \emph{International Journal of Game Theory}, 29:\penalty0 487--494, 2001.

\bibitem[Dimitrov and Sung(2007)]{DiSu07a}
Dinko Dimitrov and Shao~C. Sung.
\newblock On top responsiveness and strict core stability.
\newblock \emph{Journal of Mathematical Economics}, 43\penalty0 (2):\penalty0 130--134, 2007.

\bibitem[Dimitrov et~al.(2006)Dimitrov, Borm, Hendrickx, and Sung]{DBHS06a}
Dinko Dimitrov, Peter Borm, Ruud Hendrickx, and Shao~C. Sung.
\newblock Simple priorities and core stability in hedonic games.
\newblock \emph{Social Choice and Welfare}, 26\penalty0 (2):\penalty0 421--433, 2006.

\bibitem[Doval(2022)]{Dova22a}
Laura Doval.
\newblock Dynamically stable matching.
\newblock \emph{Theoretical Economics}, 17\penalty0 (2):\penalty0 687--724, 2022.

\bibitem[Dr{\`e}ze and Greenberg(1980)]{DrGr80a}
Jacques~H. Dr{\`e}ze and Joseph Greenberg.
\newblock Hedonic coalitions: Optimality and stability.
\newblock \emph{Econometrica}, 48\penalty0 (4):\penalty0 987--1003, 1980.

\bibitem[Elkind et~al.(2020)Elkind, Fanelli, and Flammini]{EFF20a}
Edith Elkind, Angelo Fanelli, and Michele Flammini.
\newblock Price of pareto optimality in hedonic games.
\newblock \emph{Artificial Intelligence}, 288:\penalty0 103357, 2020.

\bibitem[Ezra et~al.(2022)Ezra, Feldman, Gravin, and Tang]{EFGT22a}
Tomer Ezra, Michal Feldman, Nick Gravin, and Zhihao~Gavin Tang.
\newblock General graphs are easier than bipartite graphs: Tight bounds for secretary matching.
\newblock In \emph{Proceedings of the 22nd ACM Conference on Economics and Computation (ACM-EC)}, pages 1148 -- 1177, 2022.

\bibitem[Feldman et~al.(2009)Feldman, Korula, Mirrokni, Muthukrishnan, and P{\'a}l]{FKM+09a}
Jon Feldman, Nitish Korula, Vahab Mirrokni, Shanmugavelayutham Muthukrishnan, and Martin P{\'a}l.
\newblock Online ad assignment with free disposal.
\newblock In \emph{Proceedings of the 5th International Conference on Web and Internet Economics (WINE)}, pages 374--385, 2009.

\bibitem[Flammini et~al.(2021{\natexlab{a}})Flammini, Kodric, Monaco, and Zhang]{FKMZ21a}
Michele Flammini, Bojana Kodric, Gianpiero Monaco, and Qiang Zhang.
\newblock Strategyproof mechanisms for additively separable and fractional hedonic games.
\newblock \emph{Journal of Artificial Intelligence Research}, 70:\penalty0 1253--1279, 2021{\natexlab{a}}.

\bibitem[Flammini et~al.(2021{\natexlab{b}})Flammini, Monaco, Moscardelli, Shalom, and Zaks]{FMM+21a}
Michele Flammini, Gianpiero Monaco, Luca Moscardelli, Mordechai Shalom, and Shmuel Zaks.
\newblock On the online coalition structure generation problem.
\newblock \emph{Journal of Artificial Intelligence Research}, 72:\penalty0 1215--1250, 2021{\natexlab{b}}.

\bibitem[Flammini et~al.(2022)Flammini, Kodric, and Varricchio]{FKV22a}
Michele Flammini, Bojana Kodric, and Giovanna Varricchio.
\newblock Strategyproof mechanisms for friends and enemies games.
\newblock \emph{Artificial Intelligence}, 302:\penalty0 103610, 2022.

\bibitem[Gairing and Savani(2019)]{GaSa19a}
Martin Gairing and Rahul Savani.
\newblock Computing stable outcomes in symmetric additively separable hedonic games.
\newblock \emph{Mathematics of Operations Research}, 44\penalty0 (3):\penalty0 1101--1121, 2019.

\bibitem[Gajulapalli et~al.(2020)Gajulapalli, Liu, Mai, and Vazirani]{GLMV19a}
Karthik Gajulapalli, James Liu, Tung Mai, and Vijay~V. Vazirani.
\newblock Stability-preserving, time-efficient mechanisms for school choice in two rounds.
\newblock In \emph{Proceedings of the 40th IARCS Annual Conference on Foundations of Software Technology and Theoretical Computer Science (FSTTCS)}, 2020.

\bibitem[Gale and Shapley(1962)]{GaSh62a}
David Gale and Lloyd~S. Shapley.
\newblock College admissions and the stability of marriage.
\newblock \emph{The American Mathematical Monthly}, 69\penalty0 (1):\penalty0 9--15, 1962.

\bibitem[Huang and Tr{\"o}bst(2023)]{HuTr22a}
Zhiyi Huang and Thorben Tr{\"o}bst.
\newblock Online matching.
\newblock In Federico Echenique, Nicole Immorlica, and Vijay~V. Vazirani, editors, \emph{Online and Matching-Based Market Design}. Cambridge University Press, 2023.

\bibitem[Huang et~al.(2018)Huang, Kang, Tang, Wu, Zhang, and Zhu]{HKT+18a}
Zhiyi Huang, Ning Kang, Zhihao~Gavin Tang, Xiaowei Wu, Yuhao Zhang, and Xue Zhu.
\newblock How to match when all vertices arrive online.
\newblock In \emph{Proceedings of the 50th Annual ACM Symposium on Theory of Computing (STOC)}, pages 17--29, 2018.

\bibitem[Karp et~al.(1990)Karp, Vazirani, and Vazirani]{KVV90a}
Richard~M. Karp, Umesh~V. Vazirani, and Vijay~V. Vazirani.
\newblock An optimal algorithm for on-line bipartite matching.
\newblock In \emph{Proceedings of the 22nd Annual ACM Symposium on Theory of Computing (STOC)}, pages 352--358, 1990.

\bibitem[Kerkmann et~al.(2022)Kerkmann, Nguyen, Rey, Rey, Rothe, Schend, and Wiechers]{KNR+22a}
Anna~M. Kerkmann, N.-T. Nguyen, A.~Rey, L.~Rey, J{\"o}rg Rothe, L.~Schend, and A.~Wiechers.
\newblock Altruistic hedonic games.
\newblock \emph{Journal of Artificial Intelligence Research}, 75:\penalty0 129--169, 2022.

\bibitem[Morrill(2010)]{Morr10a}
Thayer Morrill.
\newblock The roommates problem revisited.
\newblock \emph{Journal of Economic Theory}, 145\penalty0 (5):\penalty0 1739--1756, 2010.

\bibitem[Olsen(2009)]{Olse09a}
Martin Olsen.
\newblock Nash stability in additively separable hedonic games and community structures.
\newblock \emph{Theory of Computing Systems}, 45:\penalty0 917--925, 2009.

\bibitem[Ota et~al.(2017)Ota, Barrot, Ismaili, Sakurai, and Yokoo]{OBI+17a}
Kazunori Ota, Nathana{\"e}l Barrot, Anisse Ismaili, Yuko Sakurai, and Makoto Yokoo.
\newblock Core stability in hedonic games among friends and enemies: Impact of neutrals.
\newblock In \emph{Proceedings of the 26th International Joint Conference on Artificial Intelligence (IJCAI)}, pages 359--365, 2017.

\bibitem[Pavone et~al.(2022)Pavone, Saberi, Schiffe, and Tsao]{PSST22a}
Marco Pavone, Amin Saberi, Maximilian Schiffe, and Matt~Wu Tsao.
\newblock Online hypergraph matching with delay.
\newblock \emph{Operations Research}, 70\penalty0 (4):\penalty0 2194 -- 2212, 2022.

\bibitem[Sung and Dimitrov(2010)]{SuDi10a}
Shao~C. Sung and Dinko Dimitrov.
\newblock Computational complexity in additive hedonic games.
\newblock \emph{European Journal of Operational Research}, 203\penalty0 (3):\penalty0 635--639, 2010.

\bibitem[Wang and Wong(2015)]{WaWo15a}
Yajun Wang and Sam~C. Wong.
\newblock Two-sided online bipartite matching and vertex cover: Beating the greedy algorithm.
\newblock In \emph{Proceedings of the 43rd International Colloquium on Automata, Languages, and Programming (ICALP)}, pages 1070--1081, 2015.

\bibitem[Woeginger(2013)]{Woeg13a}
Gerhard~J. Woeginger.
\newblock A hardness result for core stability in additive hedonic games.
\newblock \emph{Mathematical Social Sciences}, 65\penalty0 (2):\penalty0 101--104, 2013.

\bibitem[Yao(1977)]{Yao77a}
Andrew C.-C. Yao.
\newblock Probabilistic computations: Toward a unified measure of complexity.
\newblock In \emph{Proceedings of the 18th Symposium on Foundations of Computer Science (FOCS)}, pages 222--227. IEEE Computer Society Press, 1977.

\end{thebibliography}
\end{document}